\documentclass[final,onefignum,onetabnum]{siamart171218}
\def\webDOI{http://dx.doi.org/}

\usepackage{lipsum}
\usepackage{amsfonts}
\usepackage{graphicx}
\usepackage{epstopdf}
\usepackage{algorithmic}
\ifpdf
  \DeclareGraphicsExtensions{.eps,.pdf,.png,.jpg}
\else
  \DeclareGraphicsExtensions{.eps}
\fi

\newsiamremark{remark}{Remark}
\newsiamremark{hypothesis}{Hypothesis}
\crefname{hypothesis}{Hypothesis}{Hypotheses}
\newsiamthm{claim}{Claim}

\headers{Spectral broadening by a turbulent shear layer}{J. Garnier, E. Gay, and \'E. Savin}

\title{Multiscale analysis of spectral broadening of acoustic waves by a turbulent shear layer
\thanks{Submitted to the editors \today.
\funding{This work was partly funded by DGA/DS Mission pour la Recherche et l'Innovation Scientifique under grant 2015-60-0060.}}
}

\author{Josselin Garnier\thanks{Centre de Math\'ematiques Appliqu\'ees, \'Ecole Polytechnique, FR-91128 Palaiseau cedex, France 
  (\email{josselin.garnier@polytechnique.edu}).}
 \and Etienne Gay\thanks{ONERA/DAAA, Universit\'e Paris-Saclay, FR-92322 Ch\^atillon, France 
  (\email{etienne.gay@vo2-group.com}).}
\and \'Eric Savin\thanks{ONERA/DTIS, Universit\'e Paris-Saclay, FR-91123 Palaiseau, France 
  (\email{eric.savin@onera.fr}).}}

\usepackage{amsopn}

\usepackage{epsfig}

\usepackage{amssymb}
\usepackage{amsmath,stmaryrd}
\usepackage{color}
\usepackage{hyperref}
\usepackage{url}
\usepackage{cite}
\usepackage{nccmath}
\usepackage{cancel}

\usepackage{amsmath,amstext,amssymb}
\usepackage{mathrsfs}
\usepackage{subfigure}
\usepackage{varwidth}
\usepackage{stmaryrd}
\usepackage{bbm}
\usepackage{dsfont}

\usepackage{multimedia}
\usepackage{subeqnarray}
\usepackage{cite}
\usepackage{eurosym}
\usepackage{array}
\usepackage{pdfpages}
\usepackage{textgreek}
\usepackage{epstopdf}

\usepackage{enumitem}

\newcommand{\Lref}[1]{Lemma~\ref{#1}}  
\newcommand{\Pref}[1]{Prop.~\ref{#1}}  
\newcommand{\fref}[1]{Fig.~\ref{#1}}  
\newcommand{\eref}[1]{Eq.~\eqref{#1}} 
\newcommand{\sref}[1]{Sect.~\ref{#1}} 

\newcommand{\ci}{\mathrm{i}}

\newcommand{\rj}{r}
\newcommand{\xj}{x}

\newcommand{\zj}{z}
\newcommand{\zpj}{\zj^\prime}
\newcommand{\zuj}{\zj_1}
\newcommand{\zdj}{\zj_2}
\newcommand{\zupj}{\zj_1^\prime}
\newcommand{\zdpj}{\zj_2^\prime}
\newcommand{\rv}{{\boldsymbol\rj}}
\newcommand{\xv}{{\boldsymbol\xj}}

\newcommand{\kxj}{\kappa}
\newcommand{\kx}{{\boldsymbol\kxj}}

\newcommand{\ksp}{\kx_\text{sp}}
\newcommand{\kspu}{\kx_\text{1,sp}}

\newcommand{\kspd}{\kx_\text{2,sp}}

\newcommand{\kp}{\kx^\prime}

\newcommand{\ku}{\kx_1}

\newcommand{\kup}{\kx_1^\prime}

\newcommand{\kd}{\kx_2}

\newcommand{\kdp}{\kx_2^\prime}

\newcommand{\is}{s}
\newcommand{\zjs}{\zj_\is}
\newcommand{\xvs}{\xv_\is}

\newcommand{\om}{\omega}
\newcommand{\omref}{\om_\iref}
\newcommand{\omp}{{\om^\prime}}

\newcommand{\omu}{{\om_1}}
\newcommand{\omd}{{\om_2}}
\newcommand{\omup}{{\om_1^\prime}}
\newcommand{\omdp}{{\om_2^\prime}}

\newcommand{\qj}{q}
\newcommand{\roi}{\varrho}
\newcommand{\pres}{p}
\newcommand{\Pres}{P}
\newcommand{\cel}{c}
\newcommand{\comp}{K}

\newcommand{\rref}{\roi_\iref}
\newcommand{\pref}{\pres_\iref}
\newcommand{\celref}{\cel_\iref}
\newcommand{\compref}{\comp_\iref}

\newcommand{\pv}{{\boldsymbol\pres}}

\newcommand{\Pv}{{\boldsymbol\Pres}}

\newcommand{\MMref}{{\boldsymbol M}_\iref}

\newcommand{\uj}{u}
\newcommand{\vj}{v}
\newcommand{\ug}{{\boldsymbol\uj}}
\newcommand{\vg}{{\boldsymbol\vj}}
\newcommand{\vpertj}{V}
\newcommand{\vpert}{\boldsymbol\vpertj}

\newcommand{\vref}{\vg_\iref}
\newcommand{\vm}{\underline{\vg}}
\newcommand{\vstj}{{v}_{\text{t}}}
\newcommand{\vst}{\boldsymbol{\vstj}}
\newcommand{\Mach}{M}
\newcommand{\Macht}{\Mach_\text{t}}

\newcommand{\fj}{f}
\newcommand{\hj}{h}
\newcommand{\fv}{\boldsymbol{\fj}}

\newcommand{\Fj}{F}
\newcommand{\Hj}{H}
\newcommand{\Fv}{\boldsymbol{\Fj}}

\newcommand{\dr}{\partial}
\newcommand{\dd}{{\mathrm d}}
\newcommand{\dconv}[1]{\frac{\dd{#1}}{\dd t}}
\newcommand{\bnabla}{{\boldsymbol\nabla}}

\newcommand{\nabk}{\bnabla_\kx}
\newcommand{\Dx}{{\bf D}}

\newcommand{\di}{d}
\newcommand{\Hessian}{{\bf H}}

\newcommand{\TF}[1]{\widehat{#1}}
\newcommand{\TFk}[1]{\check{#1}}
\newcommand{\cjg}[1]{\overline{#1}}

\newcommand{\norm}[1]{|#1|}
\newcommand{\normu}[1]{\left|#1\right|}
\newcommand{\esp}[1]{{\mathbb E}\left\{#1\right\}}

\newcommand{\bet}{\beta}
\newcommand{\zet}{\zeta}

\newcommand{\BS}[1]{\boldsymbol{#1}}
\newcommand{\phase}{\phi}
\newcommand{\slowness}{\sigma}

\newcommand{\convol}{*}

\newcommand{\Id}{{\boldsymbol I}}
\newcommand{\bzero}{{\bf 0}}	
\newcommand{\iexp}{\operatorname{e}}
\newcommand{\trace}{\operatorname{Tr}}		
\newcommand{\demi}{\frac{1}{2}}

\newcommand{\iref}{0}
\newcommand{\scale}{\varepsilon}
\newcommand{\dirac}{\delta}                       
\newcommand{\Heaviside}{\operatorname{H}}                       
\newcommand{\Go}{\operatorname{O}}
\newcommand{\po}{\mathrm{o}}
\newcommand{\PSDj}{\Sigma}
\newcommand{\PSD}{{\boldsymbol\PSDj}}
\newcommand{\ACFj}{R}
\newcommand{\ACF}{{\boldsymbol\ACFj}}

\newcommand{\indic}{\mathds{1}}

\newcommand{\Rset}{\mathbb{R}}			

\newcommand{\Kj}{K}
\newcommand{\Dj}{D}
\newcommand{\Kls}{{\boldsymbol\Kj}}
\newcommand{\Dls}{{\boldsymbol\Dj}}
\newcommand{\Kms}{\Kls_\iref}

\newcommand{\fls}{\boldsymbol{s}}
\newcommand{\FLS}{\boldsymbol{S}}
\newcommand{\FLSg}{{\mathcal S}}

\newcommand{\Gr}{G}
\newcommand{\Grv}{{\boldsymbol\Gr}}

\newcommand{\Grvref}{\Grv_\iref}
\newcommand{\Grvrefm}{\Grv_\iref^+}
\newcommand{\Grvrefd}{\Grv_\iref^-}
\newcommand{\tfGrvrefm}{\TF{\Grv}_\iref^+}
\newcommand{\tfGrvrefd}{\TF{\Grv}_\iref^-}
\newcommand{\gr}{g}
\newcommand{\grv}{{\boldsymbol\gr}}

\newcommand{\tfgrvrefm}{{\TF{\grv}_\iref^+}}
\newcommand{\tfgrvrefd}{{\TF{\grv}_\iref^-}}

\newcommand{\kernel}{{\mathcal K}}
\newcommand{\dernel}{{\mathcal D}}
\newcommand{\kernelw}{\TF{\kernel}}
\newcommand{\dernelw}{\TF{\dernel}}
\newcommand{\kernelv}{\TF{\boldsymbol k}}
\newcommand{\cernelv}{\TF{\boldsymbol c}}
\newcommand{\dernelv}{\TF{\boldsymbol d}}
\newcommand{\kernelturb}{{\mathscr K}}

\newcommand{\islow}{\text{sl}}
\newcommand{\ifast}{\text{f}}
\newcommand{\iscat}{1}

\newcommand{\itr}{{\sf T}}
\newcommand{\Sj}{S}
\newcommand{\rightm}{a}
\newcommand{\leftm}{b}
\newcommand{\IO}{I_\iref}
\newcommand{\SBj}{\Psi}
\newcommand{\SB}{{\bf \SBj}}
\newcommand{\HSB}{\TF{\SB}}
\newcommand{\dist}{d}
\newcommand{\distu}{\dist_1}
\newcommand{\distd}{\dist_2}

\newcommand{\alert}[1]{\textcolor{black}{#1}}

\ifpdf
\hypersetup{
  pdftitle={Multiscale analysis of spectral broadening of acoustic waves by a turbulent shear layer},
  pdfauthor={Josselin Garnier$^{\ddagger}$, Etienne Gay$^{\dagger\star}$, \'Eric Savin$^{\dagger,\S}$ \\
$^\dagger$ ONERA \\
$^\star$ Universit\'e Paris-Diderot \\
$^\S$  CentraleSup\'elec \\
$^\ddagger$ \'Ecole Polytechnique}
}
\fi

\begin{document}

\maketitle

\begin{abstract}
We consider the scattering of acoustic waves emitted by an active source above a plane turbulent shear layer. The layer is modeled by a moving random medium with small spatial and temporal fluctuations of its mean velocity, and constant density and speed of sound. We develop a multi-scale perturbative analysis for the acoustic pressure field transmitted by the layer and derive its power spectral density when the correlation function of the velocity fluctuations is known. Our aim is to compare the proposed analytical model with some experimental results obtained for jet flows in open wind tunnels. We start with the Euler equations for an ideal fluid flow and linearize them about an ambient, unsteady inhomogeneous flow. We study the transmitted pressure field without fluctuations of the ambient flow velocity to obtain the Green's \alert{function} of the unperturbed medium with constant characteristics. Then we use a Lippmann-Schwinger equation to derive an analytical expression of the transmitted pressure field, as a function of the velocity fluctuations within the layer. Its power spectral density is subsequently computed invoking a stationary-phase argument, assuming in addition that the source is time-harmonic and the layer is thin. We finally study the influence of the source tone frequency and ambient flow velocity on the power spectral density of the transmitted pressure field and compare our results with other analytical models and experimental data.
\end{abstract}

\begin{keywords}
Aeroacoustics, shear layer, scattering, spectral broadening. 
\end{keywords}

\section{Introduction}\label{sec:Intro}

Wave propagation in random media has been extensively stu\-died in the literature. 
When the random inhomogeneities are small it can be analyzed by perturbation techniques. 
The most efficient approaches are based on techniques of separation of scales as introduced by \emph{e.g.} Papanicolaou and his coauthors \cite{ASC91}. It is found that the coherent (mean) wave amplitude decreases with the distance traveled, since wave energy is converted to incoherent fluctuations. The coherent wave experiences a deterministic spreading and a random time shift. These phenomena were originally described by O'Doherty and Anstey in a geophysical context \cite{ODA71}, and were studied mathematically in \cite{CLO94,ALF04}. The incoherent wave intensity can be calculated approximatively from a transport equation which has the form of a linear radiative transport equation \cite{ISH78,RYZ96}.

\alert{Research} on wave propagation in randomly heterogeneous media \alert{has} long been motivated by communication and imaging problems. Geosciences were the first to take an interest in imaging in random media, typically source and reflector localization \cite{BUR89,ASC91}. The objective was to better understand the interactions of waves with the medium in order to design and implement efficient imaging techniques. As the coherent wave decays exponentially with the propagation distance, the relevant information is carried by the covariance function or second-order moment of the incoherent wave field. It turns out that this covariance function was also intensively studied for the interpretation and analysis of time-reversal experiments \cite{FIN99,FOU05,FOU07}. As a result the analysis and understanding of the correlation properties of incoherent wave fields has made significant progress in the recent years and is now the basis of many modern correlation-based imaging techniques \cite{BOR03,BOR11,GAR16}.

This interest is not limited to a fixed environment, since in multiple fields the propagation of waves in a moving heterogeneous medium is an important phenomenon. In aeroacoustics for example, tests are performed in open jet wind tunnel in order to measure the acoustic signature of an active source. In these experiments acoustic waves are transmitted by some device inside the flow, and then received by microphones outside the flow. The velocity difference between the jet and the test chamber induces a turbulent shear layer, and therefore the acoustic waves interact with it and undergo several effects, such as phase and amplitude modulation. Their energy is spread over the frequency band--this is so-called spectral broadening or haystacking effect. Interactions between acoustic waves and turbulent shear layers have been studied in the literature. Candel {\it et al.} \cite{CAN75,CAN76a,CAN76b} for example conducted in the seventies a series of experiments which primarily motivated this research. These authors were particularly interested in the power spectral density (PSD) of the pressure field transmitted through the shear layer. 
Indeed, beyond its ballistic part, \alert{\emph{i.e.} the direct unscattered wave component}, the transmitted pressure field has multiply scattered components that result from the interactions with the large turbulent eddies of the shear layer.
Some properties of the shape of the induced PSD were deduced, in particular for the two lobes observed on each side of the central peak located at the emission frequency of the source. The earlier observations made by Candel {\it et al.} \cite{CAN75,CAN76a,CAN76b} on this spectral broadening effect were the starting point of more recent experiments trying to reproduce it, as for example Kr\" ober {\it et al}. \cite{KRO13} or Sijtsma {\it et al}. \cite{SIJ14}.

Analytical models for the scattered pressure field transmitted by a shear layer and its PSD were developed in \cite{GUE85,GOE01} for example. They were based on Lighthill's wave equation \cite[Section 2.2]{GOL76} derived in \cite{LIG52} with secondary source terms induced by the interaction of an incident wave field with turbulent velocity fluctuations. In \cite{CAM78a,CAM78b} mean-flow refraction effects were considered by modeling the shear layer as an oscillating vortex sheet, and assuming that turbulence induces a random phase modulation of the field. Regarding computational experiments, Ewert \textit{et al.} \cite{EWE08} numerically simulated the propagation of acoustic waves in a shear layer with a mean velocity gradient and turbulent velocity fluctuations considering the linearized Euler's equations and Pierce's convected wave equation \cite{PIE90}. The turbulent velocity fluctuations are numerically synthesized by filtering a white noise to impose prescribed statistical properties to the turbulent velocity. They are then added to the steady mean flow to form an unsteady ambient flow around which the Euler's equations are linearized. They concurrently developed an analytical model of acoustic wave scattering based on a perturbative analysis of Lilley's equation \cite[Section 1.2]{GOL76} derived in \cite{LIL72},  which was further detailed in \cite{MCA13,MCA16,POW11}. A key aspect of both approaches is a proper choice of the correlation function of the turbulent velocity fluctuations which is used to describe statistically the turbulence in the shear layer. In these works, the PSD of the acoustic pressure field transmitted by the shear layer shows the two sidebands around the emission frequency mentioned above, but some essential features fail to be reproduced by numerical simulations. Clair and Gabard \cite{CLA16} performed numerical simulations of acoustic wave scattering by a turbulent shear layer by the same methodology. A uniform mean flow was chosen. They studied in particular the influence of the source frequency, the turbulence convection velocity, and the direction of the observation \alert{(the angle between the line from the source location to the measurement location and the mean flow direction)} on the shape of the PSD of the scattered pressure. The turbulent zone is typically seen as a single large scale eddy by low frequency acoustic waves, as for scattering by a single vortex. As the frequency of the source increases, the finer turbulence eddies play an increased role in the acoustic scattering mechanism. More recently, Bennaceur {\it et al.} \cite{BEN16} reproduced numerically the spectral broadening effect by Large Eddy Simulations (LES) of the unsteady Navier-Stokes equations. Inflow velocity perturbations were added to trigger the shear layer transition to a turbulent state earlier in the simulation. The authors ended up observing a good agreement between the simulated PSD of the far-field transmitted pressure and the foregoing experiments.

\alert{In this context we propose to carry out a mathematical study of the pressure field transmitted through a plane turbulent shear layer.
The main objective of this analysis is 
to determine a simple model from first principles that gives quantitive predictions about the transmitted pressure field
in terms of the various parameters of the problem (jet velocity, correlation time and standard deviation of the velocity fluctuations, tone frequency, \emph{etc}.) and
that allows to reproduce and explain the experimental observations,
in particular, that gives an analytical formula of the PSD.
This paper extends to random flows the analysis of wave propagation phenomena in randomly stratified motionless media developed in \cite{FOU07}
and it identifies the main ingredients that are necessary to explain the observations about the spectral broadening reported in the literature.}

The paper is organized as follows.
Starting from Euler's equations, we linearize them about an ambient, unsteady inhomogeneous flow in \sref{sec:LEE}. \alert{This linearization step is essential in identifying the acoustic contributions to the overall flow and in deriving their auto-correlation properties}. In \sref{sec:no:pert} an ambient flow with constant density, speed of sound, and velocity is considered so that \alert{the expression of the Green's function of this ambient flow can be determined}. An analytical expression of the pressure field transmitted by the flow is also \alert{derived}. Then in \sref{sec:pert} we use a Lippmann-Schwinger equation to derive the pressure field transmitted by the flow with small spatial and temporal fluctuations of the ambient velocity, invoking a Born-like (or single-scattering) approximation. The PSD of the transmitted pressure is also computed in order to be able to compare the proposed model with other analytical or experimental approaches. This is the main result of the paper, which is convened in \Pref{Prop:PSD:p2}. To do so we develop a multiple scale analysis in \sref{sec:PSD} and use the stationary-phase theorem in \sref{sec:sta:pha}. A numerical example is outlined in \sref{sec:num}, before a summary of this research is finally given in \sref{sec:summary}.

\section{Linearized Euler equations about an unsteady inhomogeneous flow}\label{sec:LEE}

The full non-linear Euler equations for an ideal fluid flow in the absence of friction, heat conduction, or heat production are:
\begin{equation}\label{eq:Euler}
\begin{split}
\dconv{\roi} &+\roi\bnabla\cdot\vg=0\,, \\
\dconv{\vg} &+\frac{1}{\roi}\bnabla\pres=\bzero\,, \\
\dconv{s} &=0\,,
\end{split}
\end{equation}
where $\roi$ is the fluid density, $\vg$ is the particle velocity, $s$ is the specific entropy, and $\pres$ is the thermodynamic pressure given by the equation of state \alert{$\pres=\pres^\#(\roi,s)$, where $\pres^\#$ is a  function of the density and specific entropy}. It arises from the thermodynamic equilibrium of the fluid so it depends on its density and temperature, or entropy. Also:
\begin{equation}
\dconv{}=\frac{\dr}{\dr t}+\vg\cdot\bnabla
\end{equation}
is the usual convective derivative following the particle paths.

\subsection{Linearization of Euler equations about an ambient flow}

Linearized acoustics equations arise from the previous conservation equations when their variables are expressed as sums of ambient quantities pertaining to the ambient flow (subscript $\iref$), and lower-order acoustic fluctuations (primed quantities):
\begin{equation}
\begin{split}
\roi(\rv,t) &=\rref(\rv,t)+\roi'(\rv,t)\,, \\
\vg(\rv,t) & = \vref(\rv,t)+\vg'(\rv,t)\,, \\
s(\rv,t) & =s_\iref(\rv,t)+s'(\rv,t)\,,\\ 
\pres(\rv,t) & = \pref(\rv,t)+\pres'(\rv,t)\,,
\end{split}
\end{equation}
where $\rv=(\xv,\zj)$ is the position within the flow, and $\xv$ is the horizontal coordinates and $\zj$ is the vertical coordinate. In such a manner, the ambient quantities satisfy the Euler equations (\ref{eq:Euler}):
\begin{equation}\label{eq:ambient-cons}
\begin{split}
\dconv{\rref} &+\rref\bnabla\cdot\vref=0\,, \\
\dconv{\vref} &+\frac{1}{\rref}\bnabla\pref=\bzero\,, \\
\dconv{s_\iref} &=0\,,
\end{split}
\end{equation}
together with the following relations from the equation of state applicable to the ambient flow $\smash{\pref=\pres^\#(\rref,s_\iref)}$:
\begin{equation}\label{eq:ambient-state}
\begin{split}
\bnabla\pref &=\celref^2\bnabla\rref+\left(\frac{\dr\pres^\#}{\dr s}\right)_{\rref}\bnabla s_\iref\,, \\
\dconv{\pref} &=\celref^2\dconv{\rref}\,,
\end{split}
\end{equation}
where $\celref>0$ is the speed of sound in the ambient flow, 
$\smash{\celref^2}=\smash{( \frac{\dr\pres^\#}{\dr \roi})_{s_\iref}}$. Here and from now on one has redefined the convective derivative as:
\begin{equation}\label{eq:dconv0}
\dconv{}=\frac{\dr}{\dr t}+\vref\cdot\bnabla\,,
\end{equation}
that is, the convective derivative within the ambient flow. The primed quantities satisfy the linearized Euler equations:
\begin{equation}
\begin{split}
\dconv{\roi'} &+\roi'\bnabla\cdot\vref+\bnabla\cdot(\rref\vg')=0\,, \\
\dconv{\vg'} &+\vg'\cdot\bnabla\vref+\frac{1}{\rref}\bnabla\pres'-\frac{\roi'}{\rref^2}\bnabla\pref=\bzero\,, \\
\dconv{s'} &+\vg'\cdot\bnabla s_\iref=0\,.
\end{split}
\end{equation}
In addition, one has from the linearized equation of state:
\begin{equation}\label{eq:prime-state}
\pres'=\celref^2\roi'+\left(\frac{\dr\pres^\#}{\dr s}\right)_{\rref} s'\,.
\end{equation}

Combining \eref{eq:ambient-cons}, \eref{eq:ambient-state}, and \eref{eq:prime-state} we arrive at:
\begin{equation}
\begin{split}
\dconv{}\left(\frac{\pres'}{\rref\celref^2}\right)+\frac{1}{\rref\celref^2}\vg'\cdot\bnabla\pref+\bnabla\cdot\vg'-s'\dconv{}\left[\frac{1}{\rref\celref^2}\left(\frac{\dr\pres^\#}{\dr s}\right)_{\rref}\right] &=0\,, \\
\dconv{\vg'}+\vg'\cdot\bnabla\vref+\frac{1}{\rref}\bnabla\pres'-\frac{\pres'}{(\rref\celref)^2}\bnabla\pref-\frac{s'}{(\rref\celref)^2}\left(\frac{\dr\pres^\#}{\dr s}\right)_{\rref}\bnabla\pref &=\bzero\,, \\
\dconv{s'}+\vg'\cdot\bnabla s_\iref &=0\,.
\end{split}
\end{equation}
Discarding the last terms in $s'$ in these first two equations because they are of second order by the arguments devised in \cite{PIE90}, finally yields the system \cite{FAN01} (the last term in the second equation below is apparently missing in \cite[Eq. (2)]{FAN01})\string:
\begin{equation}\label{eq:LEEpv}
\begin{split}
\dconv{}\left(\frac{\pres'}{\compref}\right)+\bnabla\cdot\vg' +\frac{1}{\compref}\vg'\cdot\bnabla\pref &=0\,, \\
\dconv{\vg'}+\frac{1}{\rref}\bnabla\pres'+\vg'\cdot\bnabla\vref-\frac{\pres'}{\rref\compref}\bnabla\pref &=\bzero\,,
\end{split}
\end{equation}
with $\compref=\smash{\rref\celref^2}$. In terms of the unknowns $\qj:=\smash{\frac{\pres'}{\compref}}$ (dimensionless pressure) and $\ug:=\rref\vg'$ (momentum) it reads:
\begin{equation}\label{eq:LEEqu}
\begin{split}
\dconv{\qj}+\frac{1}{\rref}\bnabla\cdot\ug+\frac{1}{\rref\compref}\left(\frac{\dr\pres^\#}{\dr s}\right)_{\rref} \ug\cdot\nabla s_\iref &=0\,, \\
\dconv{\ug}+\compref\bnabla\qj+ (\trace(\Dx\vref)\Id_3+\Dx\vref)\ug + \left(\bnabla\compref-\bnabla\pref\right)\qj &=\bzero\,.
\end{split}
\end{equation}
Here $\smash{\Id}_3$ is the $3\times 3$ identity matrix, $\smash{\Dx\vref}$ stands for the velocity strain matrix within the ambient flow, such that $\smash{(\Dx\vref)_{jk}=\frac{\partial\vj_{\iref j}}{\partial\rj_k}}$. Consequently $\smash{(\Dx\vref)\ug}=\smash{\ug\cdot\bnabla\vref}$ as usually noted in fluid mechanics textbooks, and $\smash{\trace(\Dx\vref)}=\smash{\bnabla\cdot\vref}$. 
{In the experiments \cite{CAN75,CAN76a,CAN76b,KRO13,SIJ14} \alert{referred to} in \sref{sec:Intro} where no combustion occurs and temperature variations are small, the (cold) jet flow can be considered as a thermally and calorically perfect gas, whereby its heat capacities at constant volume and constant pressure $\smash{c_v}$ and $\smash{c_p}$, respectively, are constant. Therefore by the equation of state of a perfect gas $\smash{\celref^2=\gamma\frac{\pref}{\rref}}$ where $\gamma=\smash{\frac{c_p}{c_v}}$ is Laplace's coefficient,} and the above system reduces to:
\begin{equation}
\begin{split}
\dconv{\qj}+\frac{1}{\rref}\bnabla\cdot\ug &=0\,, \\
\dconv{\ug}+\compref\bnabla\qj &+ (\trace(\Dx\vref)\Id_3+\Dx\vref)\ug + \qj(\gamma-1)\bnabla\pref=\bzero\,.
\end{split}
\end{equation}
Combining the latter with \eref{eq:ambient-cons} we obtain\string:
\begin{equation}\label{Eq:lin:Eul}
\begin{split}
\dconv{\qj} &+\frac{1}{\rref}\bnabla\cdot\ug=0\,, \\
\dconv{\ug} &+\compref\bnabla\qj + (\trace(\Dx\vref)\Id_3+\Dx\vref)\ug - \qj(\gamma-1)\rref\dconv{\vref}=\bzero\,.
\end{split}
\end{equation}

\subsection{Model of the ambient flow velocity}

We define a fluctuation model of the ambient flow velocity as follows\string:
\begin{equation}\label{eq:def-flow}
\vref(t,\rv) = \left\lbrace \begin{tabular}{l l}
$\vm + \scale \vpert(t,\xv,\zj)$ & $\zj \,\in [-L,0]\,,$ \\
$\vm$ & $\text{elsewhere}\,.$ \\
\end{tabular}\right.
\end{equation}
Here $L>0$ is the width of the random flow, $\vm$ is the constant ambient flow velocity which is parallel to the horizontal coordinates $\xv$ of $\rv=(\xv,\zj)$, and $0\leq\scale\ll1$ is a small parameter which scales the amplitude of its fluctuations $\vpert$. This parameter is often called turbulence intensity in the dedicated literature, as it is related to the components of the turbulent kinetic energy. The latter are generally different for each direction \cite{CAN75,KRO13,SIJ14}, but it is assumed in the above model that such intensities are comparable for all directions. The fluctuations of the constant ambient flow velocity in \eref{eq:def-flow} are given by the mean-square stationary random vector $(\vpert(t,\rv);\,t\in\Rset,\rv\in\Rset^3)$ with zero mean:
\begin{equation}\label{eq:m_V}
\esp{\vpert(t,\xv,\zj)}=\bzero\,.
\end{equation}
Its auto-correlation matrix function reads:
\begin{multline}
\label{eq:R_V}
\esp{ \vpert(t_1,\rv_1)\otimes\vpert(t_2,\rv_2) }=\ACF_{\vpert}(t_1-t_2,\rv_1-\rv_2) \\
= \ACF(t_1 - t_2) \dirac(\xv_1 -\xv_2-(t_1-t_2)\vst)\dirac(\zuj-\zdj) \frac{\indic_{[-L,0]}(\zuj)}{L}\,,
\end{multline}
where $\smash{\vst}$ is the horizontal convection velocity of the turbulent structures (eddies) in the shear layer, and $t\mapsto\ACF(t)$ is a $3\times 3$ matrix function which describes the auto-correlation in time. Here $\zj\mapsto\smash{\indic_I(\zj)}$ is the characteristic function of the set $I$ such that $\smash{\indic_I(\zj)}=1$ if $\zj\in I$, and $0$ otherwise; and $\BS{a}\otimes\BS{b}$ stands for the usual tensor product of vectors $\BS{a}$ and $\BS{b}$ such that $\smash{(\BS{a}\otimes\BS{b})_{ij}}=\smash{a_i b_j}$ in Cartesian coordinates.
\alert{The layer is supported in the region $z \in  [-L,0]$ with $-L \leq 0$.
This model means that the fluctuations $\vpert$ are delta-correlated in space in the local frame moving at the turbulent velocity $\smash{\vst}$.
The delta function can be seen as the limit of a smooth (for instance, Gaussian) auto-correlation function when the correlation length is smaller than the acoustic wavelength.}

\begin{figure}[h!]
\centering{\includegraphics[scale=0.4]{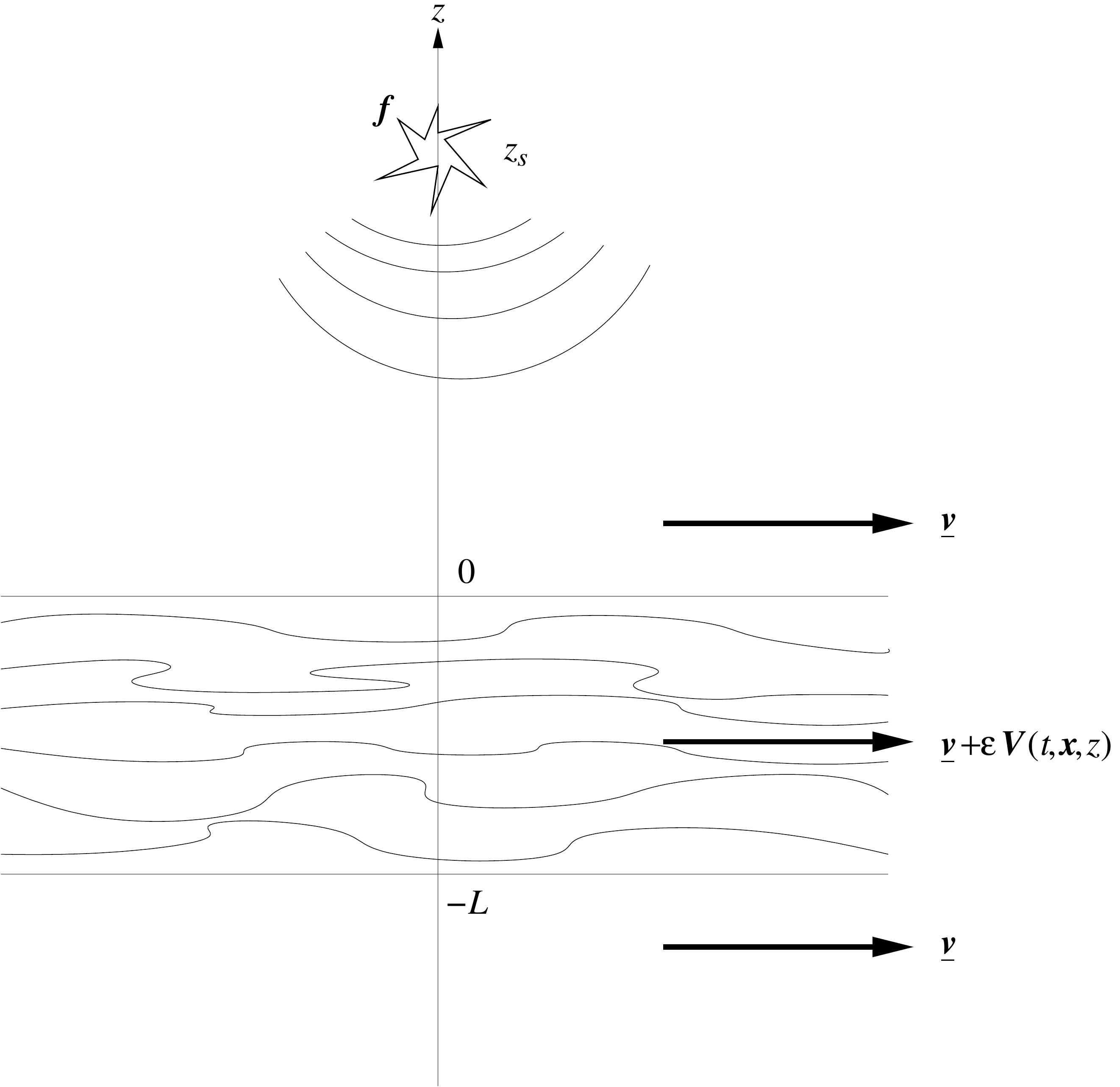}}
\caption{\textit{Acoustic waves in an homogeneous or random flow of typical thickness $L$ and ambient velocity $\vm$. The emitting sources $\fv$ are centered at some $\zjs\geq 0$.}}
\label{fg:state}
\end{figure}

\section{Transmitted fields without fluctuation of the ambient flow velocity} \label{sec:no:pert}

We start by considering the problem (\ref{Eq:lin:Eul}) without fluctuation of the ambient flow velocity. In this section we 
give the expression of the Green's function for this problem. Thus, we consider general source terms $\fls=(\hj,\fv)$ on the right hand-side of (\ref{Eq:lin:Eul}) with $\scale=0$, \emph{c.f.} (\ref{eq:def-flow}), as follows\string:
\begin{equation}\label{eq:syst:no:pert}
\begin{split}
\dconv{\qj}+\frac{1}{\rref}\bnabla\cdot\ug &=\frac{\hj}{\compref}\,, \\
\dconv{\ug}+\compref\bnabla\qj &=\rref \fv \,.
\end{split}
\end{equation}
\alert{Here $ {\hj}/{\celref^2}$ (remember that $\compref= \rref\celref^2$) corresponds to a specific mass flow rate having dimension of a density per second, and $\fv$ corresponds to an acceleration exerted on the fluid}. The convective derivative (\ref{eq:dconv0}) with the model (\ref{eq:def-flow}) is:
\begin{equation}\label{eq:dconvm}
\dconv{}=\frac{\dr}{\dr t}+\vm\cdot\bnabla\,.
\end{equation}
For the time being we consider the following specific Fourier transform and its inverse with respect to the time $t$ and the horizontal spatial coordinates $\xv$\string:
\begin{equation}\label{eq:FFTh}
\TF{\tau}(\om,\kx,\zj) =  \iint \iexp^{\ci\om(t -  \kx\cdot\xv)} \tau(t,\xv,\zj)\,\di t\di\xv\,,
\end{equation}
and\string:
\begin{equation}\label{eq:FFTh-inv}
\tau(t,\xv,\zj) = \frac{1}{(2\pi)^3}  \iint\iexp^{-\ci\om(t - \kx\cdot\xv)}\TF{\tau}(\om,\kx,\zj)\om^2 \, \di\om\di\kx\,.
\end{equation}
We note from this definition that $\kx$ is thus homogeneous to an inverse speed, or slowness. Therefore it is referred to as the (horizontal) slowness vector. Also the inverse Fourier transform with respect to that slowness vector solely will be considered in the subsequent analysis\string:
\begin{equation}\label{eq:FFTk}
\TFk{\tau}(\om,\xv,\zj)=\frac{1}{(2\pi)^2}  \int\iexp^{\ci\om\kx\cdot\xv}\TF{\tau}(\om,\kx,\zj)\,\om^2\di\kx\,.
\end{equation}
Taking the Fourier transform (\ref{eq:FFTh}) of the system (\ref{eq:syst:no:pert}) yields\string:
\begin{equation}\label{eq:FFT:syst:no:pert}
\begin{split}
-\ci\om\bet(\kx)\TF{\qj}+\frac{\ci\om}{\rref}\kx\cdot\TF{\ug}_\xv+\frac{1}{\rref}\frac{\partial\TF{\uj}_\zj}{\partial\zj} &= \frac{\TF{\hj}}{\compref}\,, \\
-\ci\om\bet(\kx)\TF{\ug}_\xv+\ci\om\compref\TF{\qj}\kx &= \rref\TF{\fv}_\xv\,, \\
-\ci\om\TF{\uj}_\zj+\compref\frac{\partial\TF{\qj}}{\partial\zj} &=\rref\TF{\fj}_\zj\,,
\end{split}
\end{equation}
with the definition ($\vm=\smash{(\vm_\xv,0)}$):
\begin{equation}\label{eq:def-beta}
\bet(\kx):=1-\kx\cdot\vm_\xv
\end{equation}
and $\smash{\TF{\fv}}=(\smash{\TF{\fv}_\xv},\smash{\TF{\fj}_\zj})$. Eliminating $\smash{\TF{\ug}_\xv}$ in the above one obtains\string:
\begin{equation}\label{eq:mean_flow_celerity}
\begin{split}
\frac{\partial\TF{\uj}_\zj}{\partial\zj}  &=\ci\om\compref\zet(\kx)^2\TF{\qj}+\frac{1}{\celref^2}\TF{\hj}+\frac{\rref}{\bet(\kx)}\kx\cdot\TF{\fv}_\xv\,, \\
\frac{\partial\TF{\qj}}{\partial\zj} &= \frac{\ci\om}{\compref}\TF{\uj}_\zj+\frac{1}{\celref^2}\TF{\fj}_\zj\,,
\end{split}
\end{equation}
where for $\bet(\kx)\neq0$\string:
\begin{equation}
\label{eq:zetm-bis}
\zet(\kx) = \sqrt{\frac{\bet(\kx)}{\celref^2} - \frac{\norm{\kx}^2}{\bet(\kx)}}
\end{equation}
is homogeneous to a slowness. 
We have thus arrived at a system of ordinary differential equations for the scalar unknowns $\smash{(\TF{\qj},\TF{\uj}_\zj)}$. Its properties may be first analyzed by considering the homogeneous differential system:
\begin{equation}
\frac{\partial}{\partial\zj}\begin{pmatrix}\TF{\qj}\\\TF{\uj}_\zj\end{pmatrix}=\ci\om\zet\begin{bmatrix} 0 & \frac{1}{\compref\zet} \\ \compref\zet & 0\end{bmatrix}\begin{pmatrix}\TF{\qj}\\\TF{\uj}_\zj\end{pmatrix}\,.
\end{equation}
Let $\smash{\IO(\kx)}=\smash{\compref\zet(\kx)}$ be the acoustic impedance and\string:
\begin{equation}
\MMref=\begin{bmatrix} \IO^\demi & \IO^{-\demi} \\ \IO^\demi & -\IO^{-\demi}\end{bmatrix}\,,\quad\begin{bmatrix} 0 & \IO^{-1} \\ \IO & 0 \end{bmatrix}=\MMref^{-1}\begin{bmatrix} 1 & 0 \\ 0 & -1 \end{bmatrix}\MMref\,.
\end{equation}
\alert{Then} the foregoing homogeneous system is diagonalized as\string:
\begin{equation}
\frac{\partial}{\partial\zj}\MMref\begin{pmatrix}\TF{\qj}\\\TF{\uj}_\zj\end{pmatrix}=\ci\om\zet\begin{bmatrix} 1 & 0 \\ 0 & -1 \end{bmatrix}\MMref\begin{pmatrix}\TF{\qj}\\\TF{\uj}_\zj\end{pmatrix}\,.
\end{equation}
\alert{This} introduces the upward and downward wave mode amplitudes $\rightm$ and $\leftm$ defined such that\string:
\begin{equation}\label{eq:wave_modes}
\MMref\begin{pmatrix}\TF{\qj}\\\TF{\uj}_\zj\end{pmatrix}=\begin{bmatrix} \iexp^{+\ci\om\zet(\kx)\zj} & 0 \\ 0 & \iexp^{-\ci\om\zet(\kx)\zj} \end{bmatrix}\begin{pmatrix}\rightm\\ \leftm\end{pmatrix}\,.
\end{equation}
The latter then satisfy\string:
\begin{equation}\label{eq:ODE_mode_hom}
\frac{\partial}{\partial\zj}\begin{pmatrix}\rightm\\ \leftm\end{pmatrix}=\bzero\,,
\end{equation}
which means that these wave mode amplitudes are independent of $\zj$ away from the source position. 
When $\zet(\kx)$ is real the wave modes are propagating,
$\rightm$ is the amplitude of the upward propagating mode,
and $\leftm$ is the amplitude of the downward propagating mode.
When $\zet(\kx)$ is imaginary the wave modes are evanescent (\emph{i.e.} they decay exponentially with the propagation distance).
Here we are interested in the far field expression of the wave so we can restrict our attention to the propagating modes.
Eq.~(\ref{eq:zetm-bis}) shows that, when $\vm={\bf 0}$, 
$\zet(\kx)$ is real if and only if $\smash{\celref\norm{\kx}}\leq 1$.
When $\vm\neq {\bf 0}$ but  the Mach number $\Mach=\smash{\norm{\vm}/\celref}<1$,
the $\kx$-domain where  $\zet(\kx)$ is real is more complicated, but it contains the disk $\smash{\celref\norm{\kx}}\leq\smash{\frac{1}{1+\Mach}}$.
We only consider subsonic flows with low Mach numbers in the subsequent developments, thus the condition $\Mach<1$ will always be fulfilled.
We can now state the main result of this section, that gives the expression of the Green's function in absence of  fluctuation of the ambient flow
velocity \alert{in terms of upward and downward wavenumber components}.

\begin{proposition}[Solution of \eref{eq:syst:no:pert} and \eref{eq:FFT:syst:no:pert}]\label{lem:no:pert}
The solution $\pv= \begin{pmatrix}  \qj \\ \ug \end{pmatrix}$ of \eref{eq:syst:no:pert} with the source term
$\fls= \begin{pmatrix}  \hj \\ \fv \end{pmatrix}$  reads:
\begin{align}
\nonumber
\pv(t,\xv,\zj) &= \Grvref\convol\fls(t,\xv,\zj)\\
&=\iiint \Grvref(t-t^\prime,\xv-\xv^\prime,\zj-\zj^\prime)\fls(t^\prime,\xv^\prime,\zj^\prime)\di t^\prime \di \xv^\prime \di \zj^\prime\,,
\end{align}
where 
\begin{equation}
\Grvref(t,\xv,\zj)=\Grvrefm(t,\xv,\zj)\Heaviside(\zj)+\Grvrefd(t,\xv,\zj)(1-\Heaviside(\zj))
\end{equation}
 is the Green's function expressed in terms of the Green's function for upward ($+$) and downward ($-$) waves
\begin{equation}
\label{def:green0}
\Grvref^\pm(t,\xv,\zj) =
\frac{\rref}{2(2\pi)^3}\iint\iexp^{\ci\om(\kx\cdot\xv\pm\zet(\kx)\zj-t)}\TF{\grv}^\pm_\iref(\kx)\otimes\TF{\grv}^\pm_\iref(\kx)\,\om^2\di\om\di\kx\,,
\end{equation}
 the Heaviside step function $\zj\mapsto\Heaviside(\zj)$ (such that $\Heaviside(\zj)=1$ if $\zj\geq 0$ and $\Heaviside(\zj)=0$ otherwise),
and the (generalized) eigenvectors of propagation:
\begin{equation}
\label{eq:eigenmodes}
\TF{\grv}^\pm_\iref(\kx)=\frac{1}{\sqrt{\zet(\kx)}} \begin{pmatrix}
{1}/{\compref} \\ {\kx}/{\bet(\kx)} \\ \pm\zet(\kx) \end{pmatrix} .
\end{equation}
For point sources located at the depth $\zjs\in[-L,0]$ with forcing terms $\FLS(t,\xv)$:
\begin{equation}
\fls(t,\xv,\zj) = \FLS(t,\xv)\dirac(\zj-\zjs)\,,
\quad 
\FLS=\begin{pmatrix} \Hj \\ \Fv_\xv \\ \Fj_\zj \end{pmatrix} , 
\end{equation}
the solution:
\begin{equation}
\pv=\begin{pmatrix}  \qj \\ \ug_\xv \\ \uj_\zj\end{pmatrix}
\end{equation}
of the system (\ref{eq:FFT:syst:no:pert}) in the Fourier domain reads for any $\zj \in (-L,\zjs)$:
\begin{equation}\label{eq:FFT-p0:gen}
\begin{split}
\TF{\pv}(\om,\kx,\zj) 
&=  \demi\iexp^{-\ci\om\zet(\kx)(\zj-\zjs)}\rref\,\TF{\FLSg}^-(\om,\kx)\tfgrvrefd(\kx)\,,
\end{split}
\end{equation}
where $\smash{\TF{\FLSg}^-}(\om,\kx)=\smash{\TF{\grv}^-_\iref(\kx)\cdot\TF{\FLS}(\om,\kx)}$.
\end{proposition}

{\it Remark.} In (\ref{def:green0}) the integral is over all $\kx$ in $\mathbb{R}^2$. 
For $\kx$ such that $\zet(\kx)$  is imaginary,
the sign of the square root of (\ref{eq:zetm-bis}) which gives rise to an exponentially decaying function in (\ref{def:green0}) is selected.
We are, however, interested only in the far-field expression, so it is possible to restrict 
the integral over the $\kx$-domain where $\zet(\kx)$ is real-valued.

\begin{proof}
We \alert{specify} the source $\fls(t,\rv)$, considering that it is located at the depth $\smash{\zjs}$ and generates forcing terms $\Fv(t,\xv)=\smash{(\Fv_\xv(t,\xv),\Fj_\zj(t,\xv))}$ and $\Hj(t,\xv)$ such that\string:
\begin{equation}
\fv(t,\xv,\zj)=\Fv(t,\xv)\dirac(\zj-\zjs)\,,\quad \hj(t,\xv,\zj)=\Hj(t,\xv)\dirac(\zj-\zjs)\,.
\end{equation}
Then the vertical momentum $\smash{\TF{\uj}_\zj}$ and pressure $\smash{\TF{\qj}}$ from \eref{eq:mean_flow_celerity} satisfy the jump conditions\string:
\begin{equation}\label{eq:jump-uq}
\begin{split}
\TF{\uj}_\zj(\om,\kx,\zjs^+)- \TF{\uj}_\zj(\om,\kx,\zjs^-) &=\frac{1}{\celref^2}\TF{\Hj}(\om,\kx)+\frac{\rref}{\bet(\kx)}\kx\cdot\TF{\Fv}_\xv(\om,\kx)\,, \\
\TF{\qj}(\om,\kx,\zjs^+)- \TF{\qj}(\om,\kx,\zjs^-) &= \frac{1}{\celref^2}\TF{\Fj}_\zj(\om,\kx)\,.
\end{split}
\end{equation}
Consequently, the upward and downward propagating mode amplitudes of \eref{eq:wave_modes} satisfy the jump conditions\string:
\begin{equation}\label{eq:jump_hom_flow}
\begin{split}
\rightm(\om,\kx,\zjs^+)- \rightm(\om,\kx,\zjs^-) &=\rref\iexp^{-\ci\om\zet(\kx)\zjs}\Sj_\rightm(\om,\kx)\,, \\
\leftm(\om,\kx,\zjs^+)- \leftm(\om,\kx,\zjs^-) &=\rref\iexp^{+\ci\om\zet(\kx)\zjs}\Sj_\leftm(\om,\kx)\,,
\end{split}
\end{equation}
with the source contributions given by\string:
\begin{equation}
\begin{pmatrix} \Sj_\rightm(\om,\kx) \\ \Sj_\leftm(\om,\kx) \end{pmatrix} = \compref^{-1}\MMref\begin{pmatrix} \TF{\Fj}_\zj(\om,\kx) \\ \TF{\Hj}(\om,\kx)+\frac{\compref}{\bet(\kx)}\kx\cdot\TF{\Fv}_\xv(\om,\kx) \end{pmatrix}\,.
\end{equation}
One thus deduces that:
\begin{equation}
\begin{pmatrix} \rightm \\ \leftm \end{pmatrix}(\om,\kx,\zj) = \rref \Heaviside(\zj-\zjs) \begin{bmatrix} \iexp^{-\ci\om\zet(\kx)\zjs} & 0 \\ 0 & \iexp^{+\ci\om\zet(\kx)\zjs} \end{bmatrix}\begin{pmatrix} \Sj_\rightm \\ \Sj_\leftm \end{pmatrix}(\om,\kx)+\begin{pmatrix} A \\ B \end{pmatrix}(\om,\kx)\,,
\end{equation}
where $A$ and $B$ are functions of $\om$ and $\kx$ independent of $\zj$ determined by the boundary conditions imposed on the wave modes. They specify the wave forms impinging the ambient flow (\ref{eq:def-flow}) at its boundaries $\zj=-L$ and $\zj=0$ if the source lies within it \cite{GAR05}, $-L<\zjs<0$, or $\zj=-L$ and $\zj=\zjs\geq 0$ if the source lies outside it \cite{FOU07}. In agreement with the experiments we have in mind, we will consider in the subsequent developments that the second situation rather takes place\string; see \fref{fg:rand_flow}. However, for the construction of the Green's function pertaining to the ambient flow, let us first consider that $-L<\zjs<0$. Assuming that no energy is coming upward from $\zj=-\infty$ or downward from $\zj=+\infty$, translates into the radiation conditions\string:
\begin{equation}
\rightm(\om,\kx,-L)=\leftm(\om,\kx,0)=0\,.
\end{equation}
Thus $A=0$ and $B=-\rref\iexp^{+\ci\om\zet(\kx)\zjs}\Sj_\leftm(\om,\kx)$, and one has with \eref{eq:wave_modes}:
\begin{align}
\nonumber
\begin{pmatrix} \TF{\qj} \\ \TF{\uj}_\zj \end{pmatrix}(\om,\kx,\zj) =
\demi \iexp^{+\ci\om\zet(\kx)(\zj-\zjs)}\begin{pmatrix} \frac{\TF{\Hj}(\om,\kx)}{\celref^2\IO(\kx)}+\frac{\rref\kx\cdot\TF{\Fv}_\xv(\om,\kx)}{\bet(\kx)\IO(\kx)}+\frac{\TF{\Fj}_\zj(\om,\kx)}{\celref^2} \\ \frac{\TF{\Hj}(\om,\kx)}{\celref^2}+\frac{\rref\kx\cdot\TF{\Fv}_\xv(\om,\kx)}{\bet(\kx)}+\frac{\IO(\kx)\TF{\Fj}_\zj(\om,\kx)}{\celref^2} \end{pmatrix}\Heaviside(\zj-\zjs) \\
+ \demi\iexp^{-\ci\om\zet(\kx)(\zj-\zjs)}\begin{pmatrix} \frac{\TF{\Hj}(\om,\kx)}{\celref^2\IO(\kx)}+\frac{\rref\kx\cdot\TF{\Fv}_\xv(\om,\kx)}{\bet(\kx)\IO(\kx)}-\frac{\TF{\Fj}_\zj(\om,\kx)}{\celref^2} \\  \frac{\IO(\kx)\TF{\Fj}_\zj(\om,\kx)}{\celref^2}-\frac{\TF{\Hj}(\om,\kx)}{\celref^2}-\frac{\rref\kx\cdot\TF{\Fv}_\xv(\om,\kx)}{\bet(\kx)} \end{pmatrix}(1-\Heaviside(\zj-\zjs))\,.
\label{eq:FTqu1}
\end{align}
From \eref{eq:FFT:syst:no:pert} one has $\smash{\TF{\ug}_\xv}=\smash{\frac{\compref}{\bet (\kx) }\TF{\qj}\kx}$ 
away from the source. As a result, after defining the four-dimensional fields $\pv$ and $\FLS$ as in the proposition, \eref{eq:FTqu1} finally yields:
\begin{align}
\nonumber
\TF{\pv}(\om,\kx,\zj)= & \;\;\demi\Heaviside(\zj-\zjs)\iexp^{+\ci\om\zet(\kx)(\zj-\zjs)}\rref\,\tfgrvrefm(\kx)\otimes\tfgrvrefm(\kx)\,\TF{\FLS}(\om,\kx) \\
&+ \demi(1-\Heaviside(\zj-\zjs))\iexp^{-\ci\om\zet(\kx)(\zj-\zjs)}\rref\,\tfgrvrefd(\kx)\otimes\tfgrvrefd(\kx)\,\TF{\FLS}(\om,\kx)\,,
\end{align}
where the (generalized) eigenvectors of propagation $\TF{\grv}^\pm_\iref(\kx)$ are defined by
(\ref{eq:eigenmodes}).
This result identifies the Green's function for upward waves $\smash{\tfGrvrefm(\om,\kx,\zj-\zjs)}$, $\smash{\zj>\zjs}$, and the Green's function for downward waves $\smash{\tfGrvrefd(\om,\kx,\zj-\zjs)}$, $\zj<\zjs$, as:
\begin{equation}
\TF{\Grv}^\pm_\iref(\om,\kx,\zj)=\demi\iexp^{\pm\ci\om\zet(\kx)\zj}\rref\,\TF{\grv}^\pm_\iref(\kx)\otimes\TF{\grv}^\pm_\iref(\kx)\,.
\end{equation}
Since $\BS{a}\otimes\BS{b}\,\BS{c}=(\BS{b}\cdot\BS{c})\BS{a}$ for any vector $\BS{c}$, $\BS{b}\cdot\BS{c}$ being the usual inner product of $\BS{b}$ and $\BS{c}$, the scalars:
\begin{align}
\nonumber
\TF{\FLSg}^\pm(\om,\kx) &=\TF{\grv}^\pm_\iref(\kx)\cdot\TF{\FLS}(\om,\kx) \\
&=\frac{1}{\sqrt{\zet(\kx)}}\left(\frac{\TF{\Hj}(\om,\kx)}{\compref}+\frac{\kx \cdot \TF{\Fv}_\xv(\om,\kx)}{\bet(\kx)} \pm\zet(\kx)\TF{\Fj}_\zj(\om,\kx)\right)
\label{eq:FLSg}
\end{align}
then turn out to be the (generalized) coordinates of the forcing terms $\FLS$ on the eigenvectors of propagation, in the Fourier domain.
\end{proof}

\begin{figure}[h!]
\centering{\includegraphics[scale=0.4]{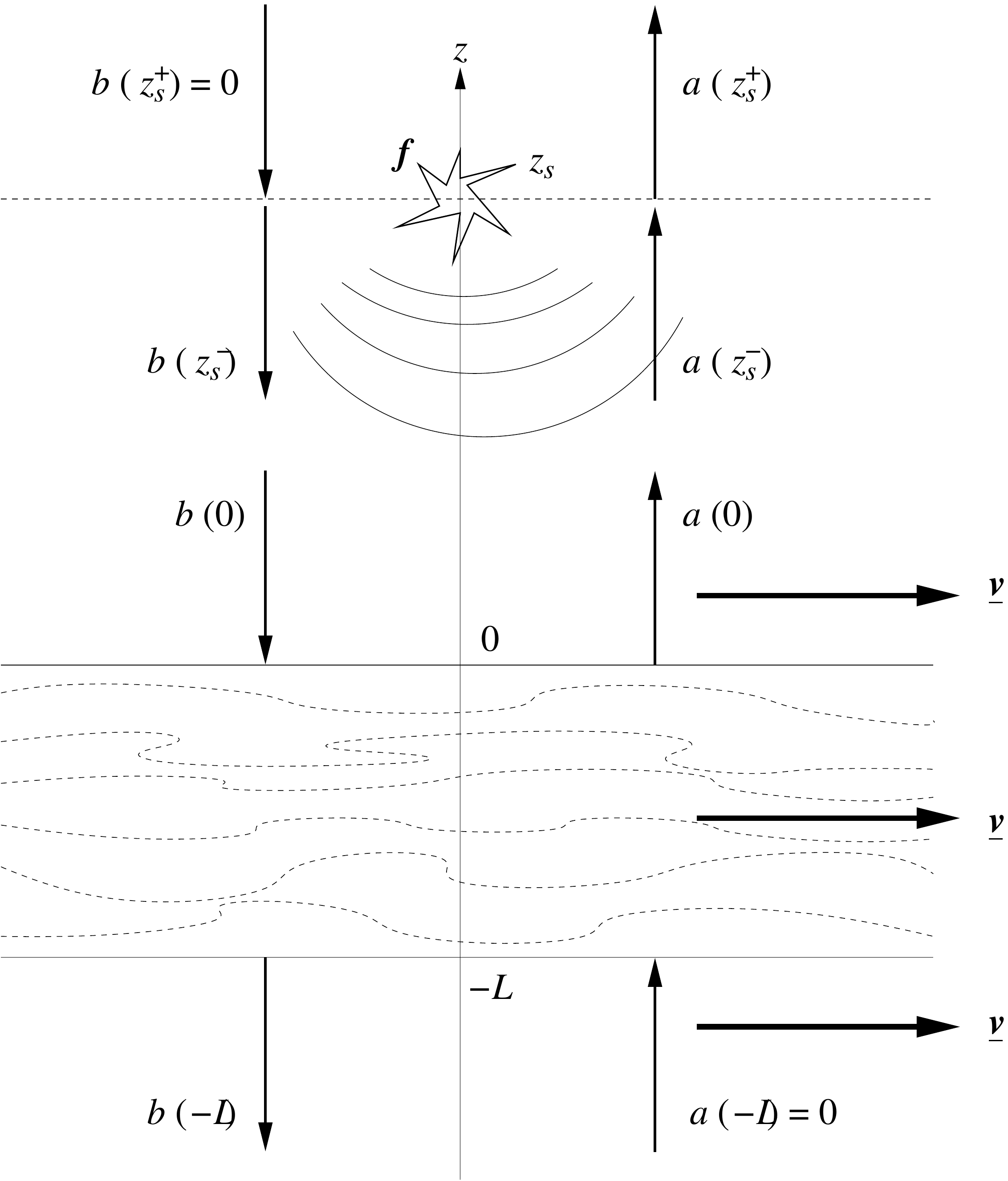}}
\caption{\textit{Acoustic waves in an homogeneous flow of typical thickness $L$ and ambient velocity $\vm$. The emitting sources $\fv$ are centered at some $\zjs\geq 0$, $\leftm$ is the downward wave mode, and $\rightm$ is the upward wave mode.}}
\label{fg:rand_flow}
\end{figure}

\section{Transmitted fields with fluctuations of the ambient flow velocity} \label{sec:pert}

We now turn to the case of an ambient flow with a fluctuating velocity, namely \eref{eq:def-flow} with $\scale > 0$. The system (\ref{Eq:lin:Eul}) reads in this case: 
\begin{equation}
\begin{split}
\dconv{\qj}+\frac{1}{\rref}\bnabla\cdot\ug &=  \frac{\hj}{\compref} - \scale\vpert\cdot\bnabla\qj\,, \\
\dconv{\ug}+\compref\bnabla\qj & =\rref\fv -\scale\vpert\cdot\bnabla\ug +\scale\rref(\gamma-1)\dconv{\vpert}\qj - \scale((\bnabla\cdot\vpert)\Id_3 +\Dx\vpert)\ug\,,
\end{split}
\end{equation}
where $\smash{\dconv{}}$ is the convective derivative (\ref{eq:dconvm}). In view of \Pref{lem:no:pert} its solution can be written as:
\begin{equation}
\begin{split}
\pv &= \Grvref\convol\left(\fls - \scale\Kls\pv\right)\,,
\end{split}
\end{equation}
where:
\begin{equation}\label{eq:Koperator}
\Kls = \begin{bmatrix} \compref\vpert\cdot\bnabla  & \bzero \\ (1-\gamma)\dconv{\vpert} & \frac{1}{\rref}(\bnabla\cdot\vpert+\vpert\cdot\bnabla)\Id_3 + \frac{1}{\rref}\Dx\vpert \end{bmatrix}\,.
\end{equation}
This is a so-called Lippmann-Schwinger equation \cite{LIP50}, of which a solution can formally be constructed by induction:
\begin{equation}
\label{eq:LS-sol}
\pv^{(n+1)} = \pv^{(0)} - \scale\Grvref\convol\Kls\pv^{(n)} \\
= \left(\Id_4 + \sum_{j=1}^n\left(-\scale\Grvref\convol\Kls\right)^j\right)\pv^{(0)}\,,
\end{equation}
with $\smash{\pv^{(0)}:=\Grvref\convol\fls}$. We will focus on the power spectral density (PSD) of $\pv$ in the subsequent developments, and more particularly the corrections to the PSD of $\smash{\pv^{(0)}}$ induced by the random fluctuations $\scale\vpert(t,\rv)$ of the constant ambient flow velocity $\vm$. The leading correction term is proportional to $\smash{\scale^2}$ as we will see below. Therefore, the above expansion is truncated at order $2$ because it is anticipated that higher order terms will be negligible:
\begin{equation}\label{eq:p2}
\pv\simeq\pv^{(2)} = \pv^{(0)} - \scale\Grvref\convol\Kls \pv^{(0)} + \scale^2\left(\Grvref\convol\Kls\right)^2 \pv^{(0)}\,.
\end{equation}
We note for convenience:
\begin{equation}
\pv^{(01)} = \Grvref\convol\Kls\pv^{(0)}\,,\quad\pv^{(02)} = \left(\Grvref\convol\Kls\right)^2\pv^{(0)}\,,
\end{equation}
the first-order and second-order terms, respectively, in the expansion of $\pv$ about the unperturbed, or ballistic pressure/momentum fields $\smash{\pv^{(0)}}$ solving \eref{eq:syst:no:pert}. \fref{fg:2nd-order-terms} sketches the zeroth, first, and second order contributions to the perturbative expansion (\ref{eq:p2}), where $\pv^{(01)}$ corresponds to the waves that have been scattered once by the random heterogeneities of the ambient flow velocity $\vm$, and $\pv^{(02)}$ corresponds to the waves that have been scattered twice. All these quantities are the combinations of upward and downward fields. We show on \fref{fg:2nd-order-terms} only the downward fields, \emph{i.e.} the transmitted fields for $\zj \leq-L$.

\begin{figure}[h!]
\centering{\includegraphics[scale=0.4]{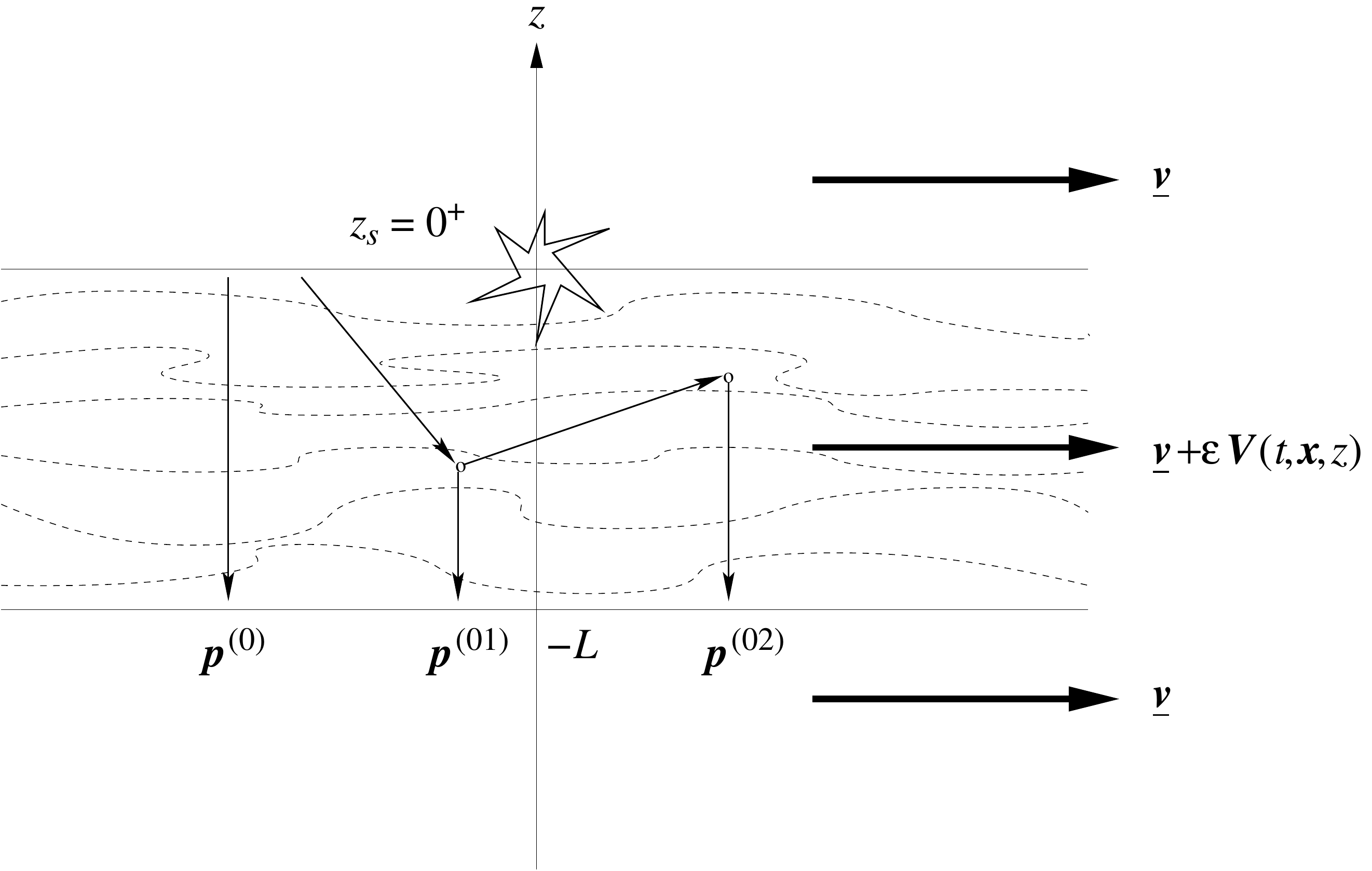}}
\caption{Zeroth (ballistic), first, and second order perturbations in the second-order expansion of the transmitted pressure/momentum fields.}\label{fg:2nd-order-terms}
\end{figure}

Regarding the first-order perturbation $\smash{\pv^{(01)}}$, we have explicitly:
\begin{equation}\label{eq:p01:TimeDom}
\pv^{(01)}(t,\xv,\zj) = \iiint \Grvref(t-t^\prime,\xv-\xv^\prime,\zj-\zj^\prime)\Kls\pv^{(0)}(t^\prime,\xv^\prime,\zj^\prime)\di t^\prime \di \xv^\prime \di \zj^\prime\,.
\end{equation}
The second-order perturbation $\smash{\pv^{(02)}}$ is expressed similarly by iterating the convolution product. As will be seen in the following \sref{sec:PSD}, the first-order contribution in this proposed perturbative model is responsible for the spectral broadening effect depicted in \cite{CAN75,CAN76a,CAN76b}. Therefore we will mainly focus on this term in the following analysis.
We can now state the main result of this section, that gives the expression of the single-scattered wave in terms of the velocity perturbations.

\begin{proposition}\label{Prop:p01}
Assume the forcing terms are point sources as in \Pref{lem:no:pert} with $\smash{\zjs=0^+}$. Then the first-order transmitted perturbations $\smash{\pv^{(01)}}$ in the perturbative expansion (\ref{eq:LS-sol}) are given by:
\begin{align}
\nonumber
\TF{\pv}^{(01)}(\om,\kx,\zj)=\frac{\rref^2\iexp^{-\ci\om\zet(\kx)\zj}}{4(2\pi)^3}\bigg(\iiint\iexp^{\ci\om\slowness(\om,\kx,\omp,\kp)\zpj}\cernelv(\om,\kx,\omp,\kp)\cdot\TF{\vpert}(\omp,\kp,\zpj) \\
\times\TF{\FLSg}^-\!\left(\om-\omp,\frac{\om\kx-\omp\kp}{\om-\omp}\right)\,\om^{\prime 2}\di\omp\di\kp\di\zpj\bigg)\,\tfgrvrefd(\kx)\,,
\label{eq:p-01-TF}
\end{align}
where $\TF{\FLSg}^-$ is the generalized coordinate given by \eref{eq:FLSg}, and the slowness $\slowness$ is given by:
\begin{equation}\label{eq:phase1}
\slowness(\om,\kx,\omp,\kp)=\zet(\kx)-\left(1-\frac{\omp}{\om}\right)\zet\left(\frac{\om\kx-\omp\kp}{\om-\omp}\right)\,.
\end{equation}
The vector $\cernelv(\om,\kx,\omp,\kp)$ is:
\begin{equation}
\label{def:hatc}
\cernelv(\om,\kx,\omp,\kp)=\kernelv(\om,\kx,\omp,\kp)-\ci\om\slowness(\om,\kx,\omp,\kp)\dernelv(\om,\kx,\omp,\kp)
\end{equation}
where $\smash{\Kms}=\smash{\mathrm{diag}(\compref,\frac{1}{\rref}\Id_3)}$, and $\kernelv(\om,\kx,\omp,\kp)$ and $\dernelv(\om,\kx,\omp,\kp)$ are given by:
\begin{multline}
\label{eq:kernelv}
\kernelv(\om,\kx,\omp,\kp)=\frac{1}{\rref\compref}\zet\left(\frac{\om\kx-\omp\kp}{\om-\omp}\right)^{-\demi}
\smash{\ci\omp\rref(\gamma-1)\bet(\kp)}
\TF{\grv}_1(\kx) \\
+\frac{1}{\rref}\TF{\grv}_1(\kx)\cdot\TF{\grv}_1\left(\frac{\om\kx-\omp\kp}{\om-\omp}\right)\begin{pmatrix} \ci\omp\kp \\ 0 \end{pmatrix}+\frac{1}{\rref}\begin{pmatrix} \ci\omp\kp \\ 0 \end{pmatrix}\cdot\TF{\grv}_1\left(\frac{\om\kx-\omp\kp}{\om-\omp}\right)\TF{\grv}_1(\kx) \\
+\tfgrvrefd(\kx)^\itr\Kms\tfgrvrefd\left(\frac{\om\kx-\omp\kp}{\om-\omp}\right)\kernelv_d\left(\om-\omp,\frac{\om\kx-\omp\kp}{\om-\omp}\right)\,,
\end{multline}
\begin{multline}
\label{eq:dernelv}
\dernelv(\om,\kx,\omp,\kp)=
\frac{1}{\rref}\TF{\grv}_1(\kx)\cdot\TF{\grv}_1\left(\frac{\om\kx-\omp\kp}{\om-\omp}\right)\begin{pmatrix} \bzero \\ 1 \end{pmatrix} \\
+\frac{1}{\rref}\begin{pmatrix} \bzero \\ 1 \end{pmatrix}\cdot\TF{\grv}_1\left(\frac{\om\kx-\omp\kp}{\om-\omp}\right)\TF{\grv}_1(\kx)\,,
\end{multline}
respectively, with $\tfgrvrefd$ defined by (\ref{eq:eigenmodes}) and
\begin{equation}
\smash{\kernelv_d(\om,\kx)}=\ci\om \begin{pmatrix} \kx \\ -\zet(\kx)\end{pmatrix}, \quad \quad 
\TF{\grv}_1(\kx)=\frac{1}{\sqrt{\zet(\kx)}}
\begin{pmatrix}
{\kx}/{\bet(\kx)}\\-\zet(\kx)\end{pmatrix} \,.
\end{equation}
\end{proposition}

\begin{proof}
 In the Fourier domain (\ref{eq:FFTh}), the first-order perturbation reads:
\begin{equation}\label{eq:TFp01}
\begin{split}
&\TF{\pv}^{(01)}(\om,\kx,z) \\
&= \iint \iexp^{\ci\om(t-\kx\cdot\xv)} \di t \di\xv \iiint \Grvref(t-t^\prime,\xv-\xv^\prime,\zj-\zj^\prime)\Kls\pv^{(0)}(t^\prime,\xv^\prime,z^\prime)\di t^\prime \di \xv^\prime \di z^\prime \\
&=  \int\!\!\left(\iint \iexp^{\ci\om(t-\kx\cdot\xv)} \Grvref(t,\xv,\zj-\zj^\prime)\di t\di\xv\right)\!\!\left(\iint \iexp^{\ci\om(t^\prime-\kx\cdot\xv^\prime)}\Kls\pv^{(0)}(t^\prime,\xv^\prime,\zj^\prime)\di t^\prime \di \xv^\prime\right) \di\zj^\prime
\end{split}
\end{equation}
with the changes of variable $t-t'\to t$ and $\xv-\xv'\to\xv$.  The first bracketed term is nothing but $\smash{\TF{\Grv}_\iref(\om,\kx,\zj-\zj^\prime)}$, while the second one is the Fourier transform of a product--hence a convolution product in the Fourier domain (\ref{eq:FFTh}). In this setting, the Fourier transform of the product of regular functions $f(t,\xv,\zj)$ and $g(t,\xv,\zj)$ reads:
\begin{equation}\label{eq:TFprod}
\TF{fg}(\om,\kx,\zj)=\frac{1}{(2\pi)^3}\iint\TF{f}(\om^\prime,\kx^\prime,\zj)\,\TF{g}\left(\om-\om^\prime,\frac{\om\kx-\om^\prime\kx^\prime}{\om-\om^\prime},\zj\right)\om^{\prime 2}\di\om^\prime\di\kx^\prime\,.
\end{equation}
Thus it remains to compute the Fourier transform (\ref{eq:FFTh}) of $\Kls$, which depends on time $t$ and the horizontal spatial coordinates $\xv$ through $\vpert(t,\xv,\zj)$ and its various products with $\bnabla=\smash{(\bnabla_\xv,\partial_\zj)}$. As for the upper left term $\smash{\Kj_{11}}=\smash{\compref\vpert\cdot\bnabla}$ of $\Kls$ for instance, we have:
\begin{multline}\label{eq:K11}
\TF{\Kj_{11}f}(\om,\kx,\zj)=\frac{\compref}{(2\pi)^3}\iint\left[\ci(\om\kx-\omp\kp)\cdot\TF{\vpert}_\xv(\omp,\kp,\zj)+\TF{\vpertj}_\zj(\omp,\kp,\zj)\partial_{\zj}\right] \\
\times\TF{f}\left(\om-\omp,\frac{\om\kx-\omp\kp}{\om-\omp},\zj\right)\om^{\prime 2}\di\omp\di\kp\,.
\end{multline}
Accordingly, the lower left term $\smash{\Kls_{21}}=\smash{(1-\gamma)\dconv{\vpert}}$ yields:
\begin{multline}\label{eq:K21}
\TF{\Kls_{21}f}(\om,\kx,\zj)=\frac{1}{\rref(2\pi)^3}\iint\TF{{\boldsymbol{\mathcal K}}}_{21}(\omp,\kp)\TF{\vpert}(\omp,\kp,\zj) \\
\times\TF{f}\left(\om-\omp,\frac{\om\kx-\omp\kp}{\om-\omp},\zj\right)\om^{\prime 2}\di\omp\di\kp\,,
\end{multline}
where 
\begin{equation}
\smash{\TF{\boldsymbol{\mathcal K}}_{21}}(\om,\kx)=\smash{\ci\om\rref(\gamma-1)\bet(\kx)\Id_3},
\end{equation}
and the remaining lower right term $\smash{\Kls_{22}}=\smash{\frac{1}{\rref}}\smash{(\bnabla\cdot\vpert \Id_3 +\vpert\cdot\bnabla  +\Dx\vpert)}$ yields:
\begin{multline}\label{eq:K22}
\TF{\Kls_{22}{\boldsymbol f}}(\om,\kx,\zj)=\frac{1}{\rref(2\pi)^3}\iint\Bigg[\left(\partial_\zj\TF{\vpertj}_\zj(\omp,\kp,\zj)
+\TF{\vpertj}_\zj(\omp,\kp,\zj)\partial_{\zj}\right)\Id_3 \\
+\ci\om\kx\cdot\TF{\vpert}_\xv(\omp,\kp,\zj)\Id_3+\TF{\vpert}(\omp,\kp,\zj)\otimes\begin{pmatrix} \ci\omp\kp \\ 0 \end{pmatrix}+\partial_\zj\TF{\vpert}(\omp,\kp,\zj)\otimes\begin{pmatrix}\bzero \\ 1\end{pmatrix}\Bigg] \\
\times\TF{{\boldsymbol f}}\left(\om-\omp,\frac{\om\kx-\omp\kp}{\om-\omp},\zj\right)\om^{\prime 2}\di\omp\di\kp\,.
\end{multline}

We finally apply the foregoing formulas to the calculation of $\smash{\TF{\pv}^{(01)}}$ in \eref{eq:TFp01}. Here $\smash{\TF{\pv}^{(0)}}$ is given by \Pref{lem:no:pert}, \eref{eq:FFT-p0:gen}, which is considered for the case $\zjs=0^+$ from now on in view of the application of \cite{CAN75,CAN76a,CAN76b} we have in mind. In this situation the Green's function $\smash{\Grvref}$ is reduced to its downward contribution $\smash{\Grvrefd}$ in the expression of $\smash{\TF{\pv}^{(0)}}$ and $\smash{\TF{\pv}^{(01)}}$, as one can see on \fref{fg:2nd-order-terms}, and $\smash{\TF{\pv}^{(0)}}$ is such that:
\begin{displaymath}
\partial_\zj\TF{\pv}^{(0)}(\om,\kx,\zj)=-\ci\om\zet(\kx)\TF{\pv}^{(0)}(\om,\kx,\zj)\,.
\end{displaymath}
This allows us to replace $\smash{\partial_\zj}$ by $\smash{-\ci(\om-\omp)\zet(\frac{\om\kx-\omp\kp}{\om-\omp})}$ in the expressions of the Fourier transforms of $\smash{\Kj_{11}\pv^{(0)}}$ and $\smash{\Kls_{22}\pv^{(0)}}$ obtained with the above formulas. 

Gathering the foregoing definitions of \eref{eq:K11}, \eref{eq:K21}, and \eref{eq:K22} one arrives at:
\begin{multline*}
\TF{\pv}^{(01)}(\om,\kx,\zj)=\frac{\rref^2\iexp^{-\ci\om\zet(\kx)\zj}}{4(2\pi)^3}\bigg(\iiint\iexp^{\ci\om\slowness(\om,\kx,\omp,\kp)\zpj}\big(\kernelw(\om,\kx,\omp,\kp,\zpj) \\
+\dernelw(\om,\kx,\omp,\kp,\zpj)\big)\,\TF{\FLSg}^-\!\left(\om-\omp,\frac{\om\kx-\omp\kp}{\om-\omp}\right)\,\om^{\prime 2}\di\omp\di\kp\di\zpj\bigg)\,\tfgrvrefd(\kx)\,,
\end{multline*}
where $\smash{\kernelw}$ and $\smash{\dernelw}$ are the scalar functions:
\begin{equation}\label{eq:kernel}
\begin{split}
\kernelw(\om,\kx,\omp,\kp,\zpj) &=\tfgrvrefd(\kx)^\itr\Kms\TF{\Kls}(\om,\kx,\omp,\kp,\zpj)\tfgrvrefd\left(\frac{\om\kx-\omp\kp}{\om-\omp}\right)\,, \\
\dernelw(\om,\kx,\omp,\kp,\zpj) &=\tfgrvrefd(\kx)^\itr\Kms\TF{\Dls}(\om,\kx,\omp,\kp,\zpj)\tfgrvrefd\left(\frac{\om\kx-\omp\kp}{\om-\omp}\right)\,,
\end{split}
\end{equation}
$\smash{\Kms=\mathrm{diag}(\compref,\frac{1}{\rref},\frac{1}{\rref},\frac{1}{\rref})}$, and $\smash{\TF{\Kls}}$ and $\smash{\TF{\Dls}}$ are the $4\times 4$ matrices:
\begin{align}
\TF{\Kls}(\om,\kx,\omp,\kp,\zpj) &= \left(\kernelv_d\left(\om-\omp,\frac{\om\kx-\omp\kp}{\om-\omp}\right)\cdot\TF{\vpert}(\omp,\kp,\zpj)\right)\Id_4 \\
\nonumber
&\quad+\begin{bmatrix} 0 & \bzero \\ \TF{\boldsymbol{\kernel}}_{21}(\omp,\kp)\TF{\vpert}(\omp,\kp,\zpj) & \TF{\boldsymbol{\kernel}}_{22}(\omp,\kp)\TF{\vpert}(\omp,\kp,\zpj)\end{bmatrix}\,, \\
\TF{\Dls}(\om,\kx,\omp,\kp,\zpj) &= \begin{bmatrix} 0 & \bzero \\ \bzero & \TF{\boldsymbol{\dernel}}_{22}(\omp,\kp)\partial_{\zpj}\TF{\vpert}(\omp,\kp,\zpj) \end{bmatrix}\,,
\end{align}
with:
\begin{align*}
\TF{\boldsymbol{\kernel}}_{22}(\omp,\kp)\TF{\vpert}(\omp,\kp,\zpj) &=\begin{pmatrix} \ci\omp\kp \\ 0 \end{pmatrix}\cdot\TF{\vpert}(\omp,\kp,\zpj)\Id_3+\TF{\vpert}(\omp,\kp,\zpj)\otimes\begin{pmatrix} \ci\omp\kp \\ 0 \end{pmatrix}\,,  \\
\TF{\boldsymbol{\dernel}}_{22}(\omp,\kp)\partial_{\zpj}\TF{\vpert}(\omp,\kp,\zpj) &=\begin{pmatrix} \bzero \\ 1 \end{pmatrix}\cdot\partial_{\zpj}\TF{\vpert}(\omp,\kp,\zpj)\Id_3+\partial_{\zpj}\TF{\vpert}(\omp,\kp,\zpj)\otimes\begin{pmatrix} \bzero \\ 1 \end{pmatrix}\,.
\end{align*}
But by a straightforward computation:
\begin{equation*}
\kernelw(\om,\kx,\omp,\kp,\zpj)=\kernelv(\om,\kx,\omp,\kp)\cdot\TF{\vpert}(\omp,\kp,\zpj) \, ,
\end{equation*}
where $\kernelv$ is given by \eref{eq:kernelv}. Likewise:
\begin{equation*}
\dernelw(\om,\kx,\omp,\kp,\zpj)=\dernelv(\om,\kx,\omp,\kp)\cdot\partial_{\zpj}\TF{\vpert}(\omp,\kp,\zpj) \, ,
\end{equation*}
where $\dernelv$ is given by \eref{eq:dernelv}. Integrating by parts in $\zpj$ one has:
\begin{multline*}
\int\iexp^{\ci\om\slowness(\om,\kx,\omp,\kp)\zpj}\dernelw(\om,\kx,\omp,\kp,\zpj)\di\zpj= \\
-\ci\om\slowness(\om,\kx,\omp,\kp)\int\iexp^{\ci\om\slowness(\om,\kx,\omp,\kp)\zpj}\dernelv(\om,\kx,\omp,\kp)\cdot\TF{\vpert}(\omp,\kp,\zpj)\di\zpj\,,
\end{multline*}
which, when combined with the foregoing expression of $\smash{\kernelw}$, gives the claimed result.
\end{proof}

\section{Computation of the power spectral density} \label{sec:PSD}

Our aim is to compare the foregoing analytical model with the measurements of \cite{CAN75,CAN76a,CAN76b}. Here the experimental results are presented in terms of the PSD or mean-square Fourier transform of the pressure field recorded at the interface of a free shear flow when an acoustic pulse is imposed at its opposite interface; see again \fref{fg:rand_flow}. The PSD is the Fourier transform of the auto-correlation function (by Wiener-Khintchin theorem). Considering the perturbative model (\ref{eq:p2}) elaborated in the previous section, the mean-square Fourier transforms are computed as:
\begin{multline}\label{eq:SB0}
\esp{\TF{\pv}^{(2)}(\omu,\ku,\zuj)\otimes\cjg{\TF{\pv}^{(2)}(\omd,\kd,\zdj)}} = \TF{\pv}^{(0)}(\omu,\ku,\zuj)\otimes\cjg{\TF{\pv}^{(0)}(\omd,\kd,\zdj)} \\
+\scale^2\TF{\pv}^{(0)}(\omu,\ku,\zuj)\otimes\esp{\cjg{\TF{\pv}^{(02)}(\omd,\kd,\zdj)}}+ \scale^2\esp{\TF{\pv}^{(02)}(\omu,\ku,\zuj)}\otimes\cjg{\TF{\pv}^{(0)}(\omd,\kd,\zdj)} \\
+\scale^2\esp{\TF{\pv}^{(01)}(\omu,\ku,\zuj)\otimes\cjg{\TF{\pv}^{(01)}(\omd,\kd,\zdj)}} \,,
\end{multline}
where $\cjg{Z}$ stands for the complex conjugate of $Z$. Indeed, $\smash{\pv^{(0)}}$ is deterministic and $\esp{\smash{\TF{\pv}^{(01)}}}=\bzero$ because by \Pref{Prop:p01} $\pv^{(01)}$ is linear with respect to $\smash{\vpert}$, which is such that $\smash{\esp{\vpert}} = \bzero$ by \eref{eq:m_V}.

From now on we assume that the sources $\FLS(t,\xv)$ are time-harmonic forcing terms emitting at the frequency $\smash{\omref}$, which means that:
\begin{equation}\label{eq:S-harmo}
\TF{\FLS}(\om,\kx)=\TF{\FLS}_\iref(\kx)\dirac(\om-\omref)\,,
\end{equation}
and therefore $\smash{\TF{\pv}^{(0)}(\om,\kx,\zj)}=\smash{\TF{\Pv}^{(0)}(\kx,\zj)\dirac(\om-\omref)}$, where $\smash{\TF{\Pv}^{(0)}}$ is deduced straightforwardly from \Pref{lem:no:pert}:
\begin{equation}\label{eq:def_P0}
\TF{\Pv}^{(0)}(\kx,\zj) = \demi\iexp^{-\ci\omref\zet(\kx)\zj}\rref\,\tfgrvrefd(\kx)\otimes\tfgrvrefd(\kx)\,\TF{\FLS}_\iref(\kx)\,.
\end{equation}
Consequently, the first, second, and third terms on the right-hand side of \eref{eq:SB0} are concentrated at the frequency $\smash{\omref}$
of the source in the frequency domain. 
The spectral broadening effect described in \cite{CAN75,CAN76a,CAN76b} should therefore be explained by the last term on the right-hand side of \eref{eq:SB0}. 
Thus this effect stems from the PSD of $\smash{\pv^{(01)}}$ with the perturbative model (\ref{eq:p2}). 
The subsequent analyses are focused on the computation of this additional contribution.

We carefully isolate the ``slow" part of the forcing terms $\smash{\TF{\FLS}_\iref(\kx)}$, which is denoted by $\smash{\TF{\FLS}_{\iref,\islow}(\kx)}$, from its ``fast" (highly oscillating) part, which is essentially a phase term $\exp(-\ci\omref\kx\cdot\xvs)$ where $\xvs$ is the horizontal central position of the sources and $\omref$ is the emitting (high) frequency of \eref{eq:S-harmo}; that is:
\begin{equation}\label{eq:S-slow}
\TF{\FLS}_\iref(\kx)=\iexp^{-\ci\omref\kx\cdot\xvs}\TF{\FLS}_{\iref,\islow}(\kx)\,.
\end{equation}
Such an ansatz will prove useful for the application of a stationary-phase argument in the following derivations. We can now state the main result of this section, which gives the expression of the covariance matrix of the single-scattered wave.

\begin{proposition}
The covariance matrix of $\smash{\TFk{\pv}^{(01)}}$:
\begin{equation}
\SB^{(01)}(\omu,\xv_1,\zuj,\omd,\xv_2,\zdj) = \esp{\TFk{\pv}^{(01)}(\omu,\xv_1,\zuj)\otimes\cjg{\TFk{\pv}^{(01)}(\omd,\xv_2,\zdj)}}
\end{equation}
along the vertical line $\xv_1=\xv_2={\bf 0}$ has the following form when $L \to 0$:
\begin{multline}\label{eq:PSD01-bis}
\SB^{(01)}(\omu,\bzero,\zuj,\omd,\bzero,\zdj) = \dirac(\omu-\omd)\frac{\rref^4(\omu-\omref)^2\om_1^4}{16(2\pi)^7} \\
\times\iiint\,\di\kx\di\ku\di\kd\iexp^{\ci\omref\phase_\iref(\ku,\kd)}\kernelturb(\omu-\omref,\kx;\omu,\ku,\omu,\kd) \\
\times\TF{\FLSg}^-_{\iref,\normalfont{\islow}}\left(\frac{\omu}{\omref}\ku+\left(1-\frac{\omu}{\omref}\right)\kx\right)\cjg{\TF{\FLSg}^-_{\iref,\normalfont{\islow}}\left(\frac{\omu}{\omref}\kd+\left(1-\frac{\omu}{\omref}\right)\kx\right)}\,\tfgrvrefd(\ku)\otimes\tfgrvrefd(\kd)\,,
\end{multline}
where $\phase_\iref$ is the overall phase of the transmitted signals:
\begin{align}
\phase_\iref( \ku,\kd)=
&-\frac{\omu}{\omref}\zuj\zet(\ku)+\frac{\omu}{\omref} \zdj \zet(\kd)-\frac{\omu}{\omref}\xvs\cdot\ku+\frac{\omu}{\omref}\xvs\cdot\kd \, ,
\label{eq:phase0}
\end{align}
$\TF{\FLSg}^-_{\iref,\normalfont{\islow}}$ is the generalized coordinate of the "slow" components of the forcing terms:
\begin{equation}
\label{eq:Sg-slow}
\TF{\FLSg}^-_{\iref,\normalfont{\islow}}(\kx)= \tfgrvrefd(\kx)\cdot\TF{\FLS}_{\iref,\normalfont{\islow}}(\kx) \, ,
\end{equation}
$\tfgrvrefd$ is defined by (\ref{eq:eigenmodes}),
the kernel $\kernelturb$  plays a fundamental in the spectral broadening:
\begin{equation}
\label{def:kernelturb}
\kernelturb(\om,\kx;\omu,\ku,\omd,\kd)=\cernelv(\omu,\ku,\om,\kx)^\itr\PSD\left(\om(1-\kx\cdot\vst)\right)\cjg{\cernelv(\omd,\kd,\om,\kx)}\,,
\end{equation}
$\cernelv$ is defined by (\ref{def:hatc}), 
and $\PSD$ is the matrix-valued spectral density of the velocity perturbations:
\begin{equation}
\label{eq:PSD_V} 
\PSD(\om)=\int\ACF(\tau)\iexp^{\ci\om\tau}\di \tau\,.
\end{equation}
\end{proposition}

\begin{proof}
We consider the covariance matrix of $\TF{\pv}^{(01)}$:
\begin{equation} \label{eq:PSD:p2}
\HSB^{(01)}(\omu,\ku,\zuj,\omd,\kd,\zdj) = \esp{\TF{\pv}^{(01)}(\omu,\ku,\zuj)\otimes\cjg{\TF{\pv}^{(01)}(\omd,\kd,\zdj)}}\,.
\end{equation}
From the result (\ref{eq:p-01-TF}) of \Pref{Prop:p01} one has:
\begin{multline}\label{eq:SB1}
\HSB^{(01)}(\omu,\ku,\zuj,\omd,\kd,\zdj) = \rref^4 \frac{\iexp^{-\ci\omu\zet(\ku)\zuj+\ci\omd\zet(\kd)\zdj}}{16(2\pi)^6} \times \\
\bigg[ \iiint\om_1^{\prime 2}\di\omup\di\kup\di\zupj\iiint\om_2^{\prime 2}\di\omdp\di\kdp\di\zdpj \,\iexp^{\ci\omu\slowness(\omu,\ku,\omup,\kup)\zupj-\ci\omd\slowness(\omd,\kd,\omdp,\kdp)\zdpj} \\
\cernelv(\omu,\ku,\omup,\kup)^\itr\esp{\TF{\vpert}(\omup,\kup,\zupj)\otimes\cjg{\TF{\vpert}(\omdp,\kdp,\zdpj)}}\cjg{\cernelv(\omd,\kd,\omdp,\kdp)} \\
\TF{\FLSg}^-\left(\omu-\omup,\frac{\omu\ku-\omup\kup}{\omu-\omup}\right)\cjg{\TF{\FLSg}^-\left(\omd-\omdp,\frac{\omd\kd-\omdp\kdp}{\omd-\omdp}\right)}\bigg]\,\tfgrvrefd(\ku)\otimes\tfgrvrefd(\kd)\,,
\end{multline}
where the slowness function $\slowness$ is given by \eref{eq:phase1}. Because the auto-correlation matrix function of $\vpert$ is given by \eref{eq:R_V}, the covariance matrix of $\smash{\TF{\vpert}}$ reads:
\begin{multline*}
\esp{\TF{\vpert}(\omu,\ku,\zuj)\otimes\cjg{\TF{\vpert}(\omd,\kd,\zdj)}} = \dirac(\zuj-\zdj) \frac{\indic_{[-L,0]}(\zuj)}{L} \\
\times\iint \ACF(t_1 - t_2) \dirac(\xv_1 - \xv_2 - (t_1-t_2)\vst) \iexp^{\ci\omu(t_1-\ku \cdot \xv_1) - \ci\omd(t_2 - \kd \cdot \xv_2)} \di t_1 \di t_2 \di \xv_1 \di \xv_2 
\end{multline*}
which by the changes of variables $\tau=\smash{t_1-t_2}$, $t=\smash{\demi(t_1+t_2)}$, $\BS{\rho}=\smash{\xv_1-\xv_2}$, and $\xv=\smash{\demi(\xv_1+\xv_2)}$ also reads:
\begin{multline}\label{esp:v}
\esp{\TF{\vpert}(\omu,\ku,\zuj)\otimes\cjg{\TF{\vpert}(\omd,\kd,\zdj)}} = \\
(2\pi)^3\dirac(\omu-\omd)\dirac(\omu\ku-\omd\kd)\dirac(\zuj-\zdj) \frac{\indic_{[-L,0]}(\zj_1)}{L}\PSD(\omu(1-\ku\cdot\vst))\,,
\end{multline}
where $\PSD(\om)$ is defined by (\ref{eq:PSD_V}).
Inserting \eref{esp:v} into \eref{eq:SB1} one arrives at the desired result 
with the change of variable $\omup \rightarrow \om$, $\kup \rightarrow \kx$.

But $\smash{\TF{\FLS}}$ is given by \eref{eq:S-harmo}, so that letting $L\to 0$ yields:
\alert{
\begin{multline*}
\HSB^{(01)}(\omu,\ku,\zuj,\omd,\kd,\zdj) = \\
\rref^4\frac{\iexp^{-\ci\omu\zet(\ku)\zuj+\ci\omd\zet(\kd)\zdj}}{16(2\pi)^3}\mathcal {I}(\omu,\ku,\omd,\kd)\tfgrvrefd(\ku)\otimes\tfgrvrefd(\kd)\,,
\end{multline*}
where:}
\begin{multline}\label{eq:I1}
\mathcal {I}(\omu,\ku,\omd,\kd) = \dirac(\omu-\omd)(\omu-\omref)^2  \int
\kernelturb(\omu-\omref,\kx;\omu,\ku,\omu,\kd) \\
\times\TF{\FLSg}^-_\iref\left(\frac{\omu}{\omref}\ku+\left(1-\frac{\omu}{\omref}\right)\kx\right)\cjg{\TF{\FLSg}^-_\iref\left(\frac{\omu}{\omref}\kd+\left(1-\frac{\omu}{\omref}\right)\kx\right)}\di\kx\,\,,
\end{multline}
and $\smash{\TF{\FLSg}^-_\iref}(\kx)=\smash{\tfgrvrefd(\kx)\cdot\TF{\FLS}_\iref(\kx)}$. Lastly, we apply the inverse Fourier transform (\ref{eq:FFTk}) with respect to $\ku$ and $\kd$ to the foregoing result in order to get an expression 
of the correlation of the transmitted fields at two points $(\xv_1,z_1)$ and $(\xv_2,z_2)$. 
The covariance matrix of $\smash{\TFk{\pv}^{(01)}}$ (remind the partial Fourier transform of \eref{eq:FFTk}) thus reads:
\begin{multline*}
\SB^{(01)}(\omu,\xv_1,\zuj,\omd,\xv_2,\zdj) = \frac{\rref^4\om_1^2\om_2^2}{16(2\pi)^7}\times \\
 \iint\iexp^{\ci\omu(\ku\cdot\xv_1-\zet(\ku)\zuj)-\ci\omd(\kd\cdot\xv_2-\zet(\kd)\zdj)}\,\mathcal {I}(\omu,\ku,\omd,\kd)\,\tfgrvrefd(\ku)\otimes\tfgrvrefd(\kd) \di\ku\di\kd\,.
\end{multline*}
Accounting for (\ref{eq:I1}) together with the ansatz (\ref{eq:S-slow}) and eventually considering 
the covariance on the vertical line $\xv_1=\xv_2=\bzero$ yields the desired result.
\end{proof}

In view of \eref{eq:SB0}, the PSD of the ballistic pressure/momentum fields $\smash{\pv^{(0)}}$ is also needed. But these quantities are, again, deterministic and time-harmonic at the frequency $\smash{\omref}$, so that one simply computes\string:
\begin{equation}\label{eq:PSD0}
\begin{split}
\SB^{(0)}({\bf 0} ,\zuj,{\bf 0},\zdj) &=\Pv^{(0)}({\bf 0},\zuj)\otimes\cjg{\Pv^{(0)}({\bf 0},\zdj)} \\
 &=\frac{\rref^2\omref^4}{4(2\pi)^4}\iint\di\ku\di\kd\iexp^{-\ci\omref[\zet(\ku)\zuj-\zet(\kd)\zdj+ \ku\cdot \xvs- \kd\cdot \xvs]} \\
&\qquad\qquad\qquad\times\TF{\FLSg}^-_{\iref,\islow}(\ku)\cjg{\TF{\FLSg}^-_{\iref,\islow}(\kd)}\,\tfgrvrefd(\ku)\otimes\tfgrvrefd(\kd)\,,
\end{split}
\end{equation}
reminding the definition (\ref{eq:def_P0}). We finally invoke a stationary-phase argument to conclude on the derivation of the PSD of the transmitted pressure field, with the foregoing expression of the phase. This last step is detailed in the next section.

\section{Stationary-phase method} \label{sec:sta:pha}

In this section we start by considering the case $\Mach=\smash{\norm{\vm}/\celref}\ll 1$ (Mach number of the ambient flow) and eventually choose $\smash{\vm}\simeq\bzero$ (\alert{$\Mach\simeq 0$ and} $\bet\simeq 1$) to simplify the derivation below. The case \alert{$\Mach\neq 0$ but small} will be addressed in subsequent developments. The first goal is to determine the stationary slowness vectors  $\smash{\kspu}$ and $\smash{\kspd}$ such that:
\begin{equation*}
\bnabla_{\ku}\phase_\iref( \kspu,\kspd)=\bnabla_{\kd}\phase_\iref( \kspu,\kspd)=\bzero\,.
\end{equation*}

\begin{lemma}
In the instance that $\smash{\omu}$ and $\smash{\omd}$ are equal to $\smash{\omref}$ at the first order, namely $1-\smash{\frac{\omu}{\omref}}=\po(1)$ and $1-\smash{\frac{\omd}{\omref}}=\po(1)$, one has:
\begin{equation}\label{lem:ksp12}
\normalfont{\kspu}\simeq \zet(\normalfont{\kspu})\frac{\xvs}{\zuj}=\frac{\xvs}{\celref\sqrt{\norm{\xvs}^2+\zuj^2}}\,, \quad \quad
\normalfont{\kspd}\simeq \zet(\normalfont{\kspd})\frac{\xvs}{\zdj}=\frac{\xvs}{\celref\sqrt{\norm{\xvs}^2+\zdj^2}}\,.
\end{equation}
\end{lemma}

\begin{proof}
$\zet(\kx)$ is given by \eref{eq:zetm-bis} with $\bet(\kx)=1$, therefore $\nabk\zet(\kx)=-\kx/\zet(\kx)$. Thus:
\begin{equation}
\label{eq:grad-phi}
\bnabla_{\ku}\phase_\iref(\ku,\kd) = \frac{\omu}{\omref}\left[\zuj \frac{\ku}{\zet(\ku)}-\xvs\right]\,, \quad \quad 
\bnabla_{\kd}\phase_\iref(\ku,\kd) = -\frac{\omu}{\omref}\left[\zdj \frac{\kd}{\zet(\kd)}-\xvs\right] \,.
\end{equation}
Then for $\smash{\frac{\omu}{\omref}}=\Go(1)$ one obtains the claimed result.
\end{proof}

We subsequently apply the stationary-phase theorem to \eref{eq:PSD01-bis} to obtain the final expression of the PSD of the scattered transmitted pressure field given in \Pref{Prop:PSD:p2}, which is the main result of the paper.

\begin{proposition} \label{Prop:PSD:p2}
Let $\dist=\smash{(\norm{\xvs}^2+\zj^2)^\demi}$ be the distance from the time-harmonic sound source (\ref{eq:S-harmo}) at $(\xvs,0^+)$ (with $\smash{\TF{\FLS}_\iref}=\smash{(\TF{\Hj}_\iref,\TF{\Fv}_{\iref\xv},\TF{\Fj}_{\iref\zj})}$) to the observation point $(\bzero,\zj)$ at the depth $\zj$ ($\zj<-L$) on the outer side of the flow. Then the ballistic transmitted pressure field $\smash{\pres'^{(0)}}=\smash{\compref\qj^{(0)}}$ has the form:
\begin{equation}
\TFk{\pres}'^{(0)}(\om,\bzero,\zj)\cjg{\TFk{\pres}'^{(0)}(\omp,\bzero,\zj)}=\delta(\om-\omref)\delta(\omp-\omref)\SBj_\iref(\xvs,\zj,\omref) \, ,
\end{equation}
with\string:
\begin{equation}\label{eq:PSD0final}
\SBj_0(\xvs,\zj,\omref)=\frac{\rref^2\omref^2}{4(2\pi)^2\celref^2\dist^4}\normu{\frac{\dist}{\rref\celref}
\TF{\Hj}_\iref\left(\frac{\xvs}{\celref\dist}\right)+\xvs\cdot\TF{\Fv}_{\iref\xv}\left(\frac{\xvs}{\celref\dist}\right)
-\zj\TF{\Fj}_{\iref\zj}\left(\frac{\xvs}{\celref\dist}\right)}^2\,,
\end{equation}
and for $\smash{1-\frac{\om}{\omref}}=\po(1)$ the mean-square Fourier transform of the scattered transmitted pressure field $\pres'^{(01)}=\smash{\compref\qj^{(01)}}$ reads:  
\begin{equation}
\esp{\TFk{\pres}'^{(01)}(\om,\bzero,\zj)\cjg{\TFk{\pres}'^{(01)}(\omp,\bzero,\zj)}}=\delta(\om-\omp)\SBj_{01}(\xvs,\zj,\om,\omref) \, ,
\end{equation}
where:
\begin{equation}
\label{eq:PSDfinal}
\SBj_{01}(\xvs,\zj,\om,\omref)=\frac{\rref^2\omref^2}{4(2\pi)^3}\left(1-\frac{\om}{\omref}\right)^2\left(\frac{\om}{\omref}\right)^4\,\SBj_\iscat(\xvs,\zj,\om,\omref)\SBj_\iref\left(\xvs,\zj,\omref\right)\,,
\end{equation}
and $\smash{\SBj_\iscat(\xvs,\zj,\om,\omref)}$ is a filtered frequency spectrum responsible for the spectral broadening of the source centered about $\omref$ given by:
\begin{equation}\label{eq:Sigmaw_iscat}
\SBj_\iscat(\xvs,\zj,\om,\omref)=\int_{\norm{\kx}\leq\frac{1}{\celref}} \kernelturb\left(\om-\omref,\kx;\om, \frac{\xvs}{\celref\dist},\om,\frac{\xvs}{\celref\dist}\right)\,\di\kx\,,
\end{equation}
where $\kernelturb$ is defined by (\ref{def:kernelturb}).
\end{proposition}

\begin{proof}
Expanding the phase (\ref{eq:phase0}) in a Taylor series about the small increment $\smash{\omref-\omu}$ yields:
\begin{equation}
\phase_\iref(\ku,\kd) \simeq -\frac{\omu}{\omref} \zuj \zet(\ku)+\frac{\omu}{\omref} \zdj \zet(\kd) -\frac{\omu}{\omref}\left(\ku-\kd\right)\cdot\xvs \,.
\end{equation}
We can thus identify the slow part of the phase $\smash{\phase_\islow}$ from its fast part $\smash{\phase_\ifast}$ proportional to $\omu/\omref$ as follows:
\begin{equation}
\begin{split}
\phase_\ifast(\ku,\kd) =&-\zet(\ku)\zuj+\zet(\kd)\zdj-(\ku-\kd)\cdot\xvs \,, \\
\phase_\islow(\ku,\kd) =&\left(1-\frac{\omu}{\omref}\right)\left( \zuj \zet(\ku) +\ku\cdot\xvs- \zdj \zet(\kd) -\kd\cdot\xvs\right)\,.
\end{split}
\end{equation}
Applying the stationary-phase theorem in dimension $n=4$ (see for example \cite[Eq. (14.77)]{FOU07}) to the auto-spectrum (\ref{eq:PSD01-bis}) of the first-order perturbations of the transmitted pressure/momentum fields yields:
\begin{multline*} \label{eq:inter:psd:p2}
\SB^{(01)}(\omu,\bzero,\zuj,\omd,\bzero,\zdj) = \dirac(\omu-\omd)\frac{\rref^4(\omu-\omref)^2\om_1^4}{16\omref^2(2\pi)^5}\frac{\iexp^{\ci(n^*-2)\frac{\pi}{2}}\iexp^{\ci\omref\phase_\ifast(\kspu,\kspd)}}{\sqrt{\norm{\det\Hessian_{\ku,\kd}\phase_\ifast(\kspu,\kspd)}}} \\
\times \iexp^{\ci\omref\phase_\islow(\kspu,\kspd)} \int\kernelturb(\omu-\omref,\kx;\omu,\kspu,\omu,\kspd)\di\kx \\
\times\TF{\FLSg}^-_{\iref,\islow}\left(\kspu\right)\cjg{\TF{\FLSg}^-_{\iref,\islow}\left(\kspd\right)}\,\tfgrvrefd(\kspu)\otimes\tfgrvrefd(\kspd)\,,
\end{multline*}
where $n^*$ is the number of positive eigenvalues of the Hessian $\smash{\Hessian_{\ku,\kd}\phase_\ifast(\kspu,\kspd)}$. The latter is block-diagonal owing to \eref{eq:grad-phi} and reads:
\begin{equation*}
\Hessian_{\ku,\kd}\phase_\ifast(\ku,\kd)=\begin{bmatrix} \frac{\zuj}{\zet(\ku)}\left(\Id_2+\frac{\ku\otimes\ku}{\zet(\ku)^2}\right) & \bzero \\ \bzero & -\frac{\zdj}{\zet(\kd)}\left(\Id_2+\frac{\kd\otimes\kd}{\zet(\kd)^2}\right) \end{bmatrix}\,,
\end{equation*}
such that\string:
\begin{equation*}
\det\Hessian_{\ku,\kd}\phase_\ifast(\ku,\kd)=\left(\frac{\zuj\zdj}{\zet(\ku)\zet(\kd)}\right)^2\left(1+\frac{\norm{\ku}^2}{\zet(\ku)^2}\right)\left(1+\frac{\norm{\kd}^2}{\zet(\kd)^2}\right)\,.
\end{equation*} 
From this we can deduce that the eigenvalues of $\smash{\Hessian_{\ku,\kd}\phase_\ifast(\ku,\kd)}$ are:
\begin{equation*}
\begin{split}
\lambda_1 &= \frac{\zuj}{\zet(\ku)}\,,\quad\lambda_2= \frac{\zuj}{\zet(\ku)}\left(1 + \frac{\norm{\ku}^2}{\zet(\ku)^2}\right)\,, \\
\lambda_3 &= -\frac{\zdj}{\zet(\kd)}\,,\quad\lambda_4= -\frac{\zdj}{\zet(\kd)}\left(1+ \frac{\norm{\kd}^2}{\zet(\kd)^2}\right)\,,
\end{split}
\end{equation*}
and therefore $n^*=2$. Besides:
\begin{equation*}
\sqrt{\norm{\det\Hessian_{\ku,\kd}\phase_\ifast(\kspu,\kspd)}}=\frac{\celref^2\distu^2\distd^2}{\zuj\zdj}
\end{equation*}
introducing $\distu=\smash{(\norm{\xvs}^2+\zuj^2)^\demi}$ and $\distd=\smash{(\norm{\xvs}^2+\zdj^2)^\demi}$, and:
\begin{equation*}
\phase_\ifast(\kspu,\kspd)=\frac{\distd-\distu}{\celref}\,.
\end{equation*}
Indeed, owing to \Lref{lem:ksp12} one also has:
\begin{equation}
\zet(\kspu)=\frac{\zj_1}{\celref\distu}\,,\quad\zet(\kspd)=\frac{\zj_2}{\celref\distd}\,.
\end{equation}
We can finally compute the mean-square Fourier transform of the first-order transmitted pressure field $\smash{\pres'^{(01)}}=\smash{\compref\qj^{(01)}}$ 
as the upper-left term of $\smash{\SB^{(01)}(\omu,\bzero,\zuj,\omd,\bzero,\zdj)}$, and that of the ballistic transmitted pressure field $\smash{\pres'^{(0)}}=\smash{\compref\qj^{(0)}}$ as the upper-left term of $\smash{\SB^{(0)}(\bzero,\zuj,\bzero,\zdj)}$. Using \eref{eq:eigenmodes} with $\bet(\kx)=1$ and \eref{eq:FLSg} with the definition (\ref{eq:S-slow}), we have:
\begin{equation*}
\TF{\FLSg}^-_{\iref,\islow}(\kspu)=\smash{\iexp^{\ci\frac{\omref}{\celref}\frac{\norm{\xvs}^2}{\distu}}\tfgrvrefd(\kspu)\cdot\TF{\FLS}_\iref(\kspu)}  \, ,
\end{equation*}
with a similar expression for $\smash{\TF{\FLSg}^-_{\iref,\islow}(\kspd)}$; see \eref{eq:Sg-slow}. We end up with:
\begin{multline*}
\TFk{\pres}'^{(0)}(\om,\bzero,\zuj)\cjg{\TFk{\pres}'^{(0)}(\omp,\bzero,\zdj)}=\delta(\om-\omref)\delta(\omp-\omref)\frac{\iexp^{\ci\frac{\omref}{\celref}(\frac{\zdj^2}{\distd}-\frac{\zuj^2}{\distu})}}{4(2\pi)^2}\frac{\rref^2\omref^2}{\distu\distd} \\
\times\varphi\left(\frac{\xvs}{\celref\distu},\frac{\zuj}{\celref\distu}\right)\cjg{\varphi\left(\frac{\xvs}{\celref\distd},\frac{\zdj}{\celref\distd}\right)}\,,
\end{multline*}
and:
\begin{multline*}
\esp{\TFk{\pres}'^{(01)}(\omu,\bzero,\zuj)\cjg{\TFk{\pres}'^{(01)}(\omd,\bzero,\zdj)}}= \delta(\omu-\omd)\frac{\iexp^{\ci\frac{\omref}{\celref}(\frac{\zdj^2}{\distd}-\frac{\zuj^2}{\distu})}}{16(2\pi)^5}\frac{\rref^4\omref^4}{\distu\distd} \\
\times\left(1-\frac{\omu}{\omref}\right)^2\left(\frac{\omu}{\omref}\right)^4\SBj(\xvs,\zuj,\zdj,\omu,\omref)\varphi\left(\frac{\xvs}{\celref\distu},\frac{\zuj}{\celref\distu}\right)\cjg{\varphi\left(\frac{\xvs}{\celref\distd},\frac{\zdj}{\celref\distd}\right)}\,,
\end{multline*}
where:
\begin{multline}
\label{eq:Sigmaw}
\SBj(\xvs,\zuj,\zdj,\om,\omref)= \iexp^{\ci\omref\phase_\islow(\kspu,\kspd)}\\
\times 
\int_{\norm{\kx}\leq\frac{1}{\celref}}\kernelturb(\om-\omref,\kx;\om,\kspu,\om,\kspd)\,\di\kx\,,
\end{multline}
and, for $\smash{\TF{\FLS}_\iref}=\smash{(\TF{\Hj}_\iref,\TF{\Fv}_{\iref\xv},\TF{\Fj}_{\iref\zj})}$:
\begin{equation*}
\varphi(\kx,\zet)=\frac{1}{\compref}\TF{\Hj}_\iref(\kx)+\kx\cdot\TF{\Fv}_{\iref\xv}(\kx)-\zet\TF{\Fj}_{\iref\zj}(\kx)\,.
\end{equation*}
We thus define:
\begin{displaymath}
\SBj_\iref(\xvs,\zj,\omref) =\frac{\rref^2\omref^2}{4(2\pi)^2\dist^2}\normu{\varphi\left(\frac{\xvs}{\celref\dist},\frac{\zj}{\celref\dist}\right)}^2
\end{displaymath}
for $\dist=\smash{(\norm{\xvs}^2+\zj^2)^\demi}$, and:
$$
\SBj_\iscat(\xvs,\zj,\omu,\omref)=\SBj(\xvs,\zj,\zj,\omu,\omref)=
\int_{\norm{\kx}\leq\frac{1}{\celref}} \kernelturb(\omu-\omref,\kx;\omu,\kspu,\omu,\kspu)\,\di\kx\,,
$$
then we obtain the claimed formulas (\ref{eq:PSD0final}) and (\ref{eq:PSDfinal}).
\end{proof}

It remains to compute the integral in Eq.~(\ref{eq:Sigmaw_iscat}).
Assuming that $\ACF(\tau)$ is diagonal, namely $\ACF(\tau)=\ACFj(\tau)\Id_3$, and denoting by $\PSDj$ the Fourier transform of $\ACFj$, one has:
\begin{multline*}
\kernelturb(\om-\omref,\kx;\om,\ksp,\om,\ksp) =  \\
\PSDj((\om-\omref)(1-\vst\cdot\kx))\cernelv(\om,\ksp,\om-\omref,\kx)\cdot\cjg{\cernelv(\om,\ksp,\om-\omref,\kx)}\,,
\end{multline*}
with $\ksp= {\xvs}/({\celref\dist})$.
These results are illustrated in the next section for a correlation function and parameters adapted from the experiments in \cite{CAN75} and analytical models in \cite{MCA16}. 

\section{Numerical example}\label{sec:num}

At first, we choose a Gaussian model for the autocorrelation function $\tau\mapsto\ACF(\tau)$ of \eref{eq:R_V}:
\begin{equation}
\ACF(\tau)=\ACFj(\tau)\Id_3, \quad \quad 
\ACFj(\tau)=\sigma_\vpertj^2\exp\left(-\pi\frac{\tau^2}{4\tau_s^2}\right)\,.
\end{equation}
Here $\tau_s$ is the correlation time, or the turbulence integral timescale in the dedicated literature, \emph{i.e.} the typical time scale of a realization of the turbulent velocity fluctuations $\vpert$ such that $\smash{\int_0^{+\infty}{\ACFj(\tau)}/{\ACFj(0)}\di\tau}=\smash{\tau_s}$; and $\smash{\sigma_\vpertj}$ quantifies their standard deviation. Then from \eref{eq:PSD_V}:
\begin{equation}
\PSDj(\om)=2\tau_s\sigma_\vpertj^2\exp\left(-\frac{1}{\pi}\tau_s^2\om^2\right)\,.
\end{equation}
In view of \eref{eq:R_V} and \eref{eq:def-flow}, the variance $\smash{\sigma_\vpertj^2}$ scales as a squared velocity--typically $\smash{\norm{\vm}^2}$ the squared mean jet velocity. In order to compare our results of \Pref{Prop:PSD:p2} with the experimental results in \cite{CAN75} and the analytical models in \emph{e.g.} \cite{MCA13,MCA16,POW11}, we plot on \fref{fg:DSP-vs-w0} and \fref{fg:DSP-vs-vt} the normalized ``power spectrum" $\smash{\SBj_2(\xvs,\zj,\omref,\om)}$ of the transmitted pressure field defined by:
\begin{equation}
\SBj_2(\xvs,\zj,\om,\omref)=\delta(\om-\omref)+\frac{\scale^2\rref^2}{4(2\pi)^3}(\om-\omref)^2\left(\frac{\om}{\omref}\right)^4\,\SBj_\iscat(\xvs,\zj,\om,\omref)
\end{equation}
in dB ($\smash{\SBj_2^\text{dB}}=\smash{10\log_{10}\SBj_2}$) for the data provided in \cite{MCA09,MCA16}: $L=0.1$ m for the thickness of the turbulent shear layer at a distance $D=0.5$ m from the jet in the experiments of Candel \emph{et al.} \cite{CAN75}, $\smash{\tau_s}=\smash{L/\norm{\vst}}$, $\scale=12\%$ for the turbulence intensity, $\norm{\vst}=0.5\norm{\vm}$ for the velocity of the turbulent eddies, and $\celref=340$ m/s and $\rref=1.2$ kg/m$^3$ for the ambient flow characteristics. More particularly, \fref{fg:DSP-vs-w0} shows $\smash{\SBj_2(\bzero,-L,\omref,\om)}$ as a function of $\smash{\Delta\om=\om-\omref}$ for various tone frequencies $\smash{f_\iref= {\omref}/({2\pi})}$ in the range [6--20] kHz and the jet velocity $\smash{U_J}=\norm{\vm}=60$ m/s, and \fref{fg:DSP-vs-vt} shows $\smash{\SBj_2(\bzero,-L,\omref,\om)}$ as a function of $\Delta\om$ for various jet velocities $\smash{U_J}$ in the range [20--60] m/s and $\smash{f_\iref}=20$ kHz. \alert{In addition, we show in these figures the analytical results displayed in \cite{MCA16} for the homogeneous axisymmetric turbulence (HAT), and the three-dimensional Gaussian homogeneous isotropic turbulence (HIT) correlation functions developed in this work. Both models are based on the general model of isotropic turbulence correlation derived in \cite{BAT53,KAR38,ROB40}. The HAT model takes explicitly into account the axisymmetry of a circular jet. Also the experiments in \cite{CAN75} are for a circular jet with a large radius, so that the plane shear layer model developed here is applicable for a source at the center of the jet}. 

The two lobes \alert{of the experimental results shown in \fref{fg:DSP-vs-w0}(a) and \fref{fg:DSP-vs-vt}(a)} are well recovered, together with their position on the frequency axis which has been observed to be independent of the tone frequency $\smash{f_\iref}$. Indeed, from the expression of $\SBj_{01}(\xvs,\zj,\omref,\om)$ in \Pref{Prop:PSD:p2} the maxima of the lobes are approximately found as the maxima of $\smash{\Delta\om^2\exp(-\frac{1}{\pi}\tau_s^2\Delta\om^2(1-\Macht)^2)}$ which yields:
\begin{displaymath}
\norm{\Delta\om_\text{max}}\simeq\sqrt{\pi}\frac{\celref}{L}\Macht\,,
\end{displaymath}
where $\Macht=\norm{\vst}/\celref\ll 1$ is the (small) Mach number for the turbulent eddies. This estimate is in good agreement with \cite{BEN16,CAN75,CAN76a,CAN76b,CLA16,MCA09,MCA13,POW11,SIJ14}. Despite several simplifications, the proposed model allows to recover the main trends of the power spectrum of the transmitted pressure field already outlined in those previous works, namely:
\begin{itemize}
\item a linear growth in the position of the maximum of the lobes as a function of the velocity of the turbulent eddies $\Macht$. Also the width of the lobes and thus the amount of scattered energy increases as well when $\Macht$ increases;
\item the energetic macro-eddies contribute to the larger part of the scattered pressure field. Spectral broadening is related to a Doppler shift due to the motion of these large structures which act as secondary sources for the scattered field.
\end{itemize}

\begin{figure}[h!]
\centering     
\subfigure[]{\label{1}\includegraphics[scale=0.35]{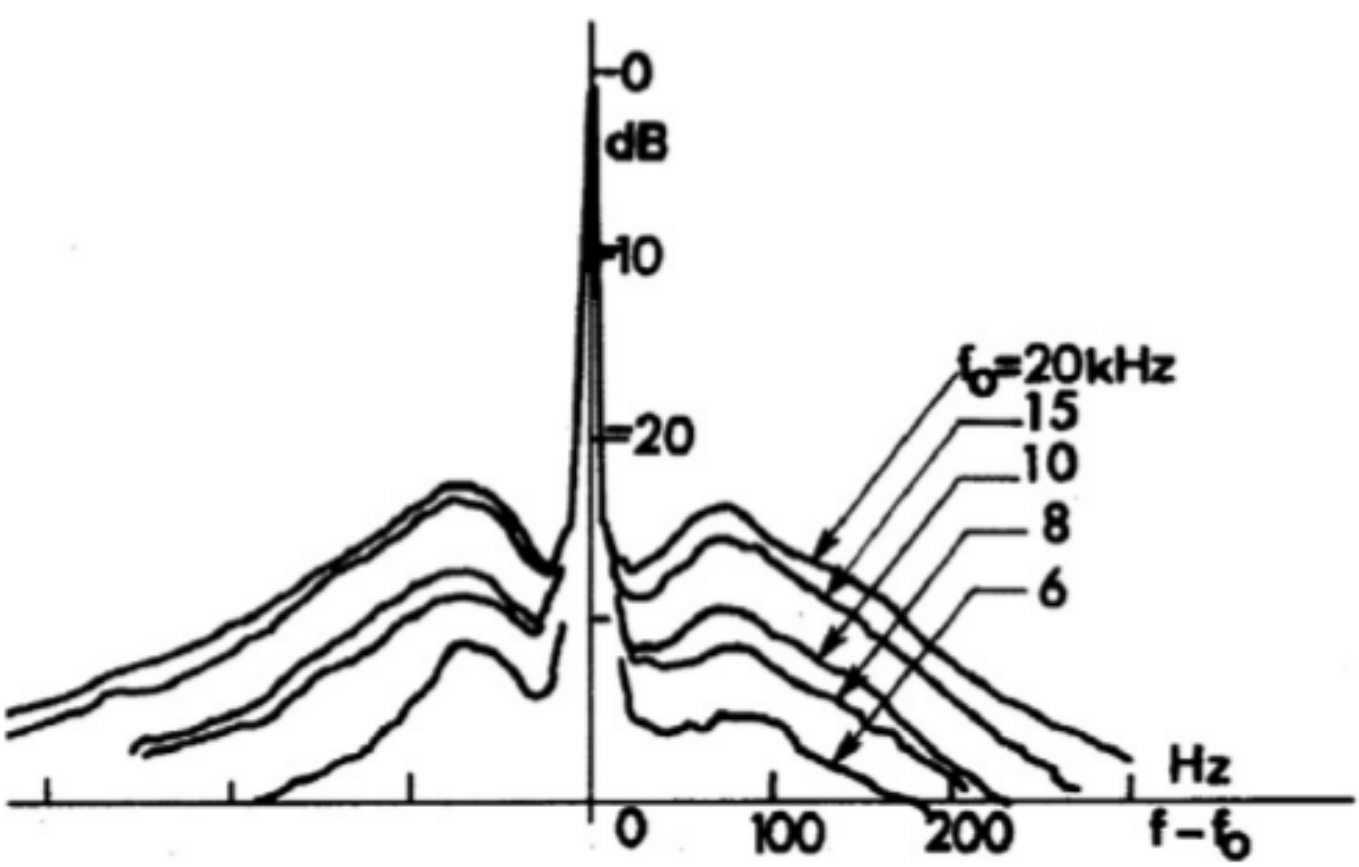}}
\subfigure[]{\label{2}\includegraphics[scale=0.27]{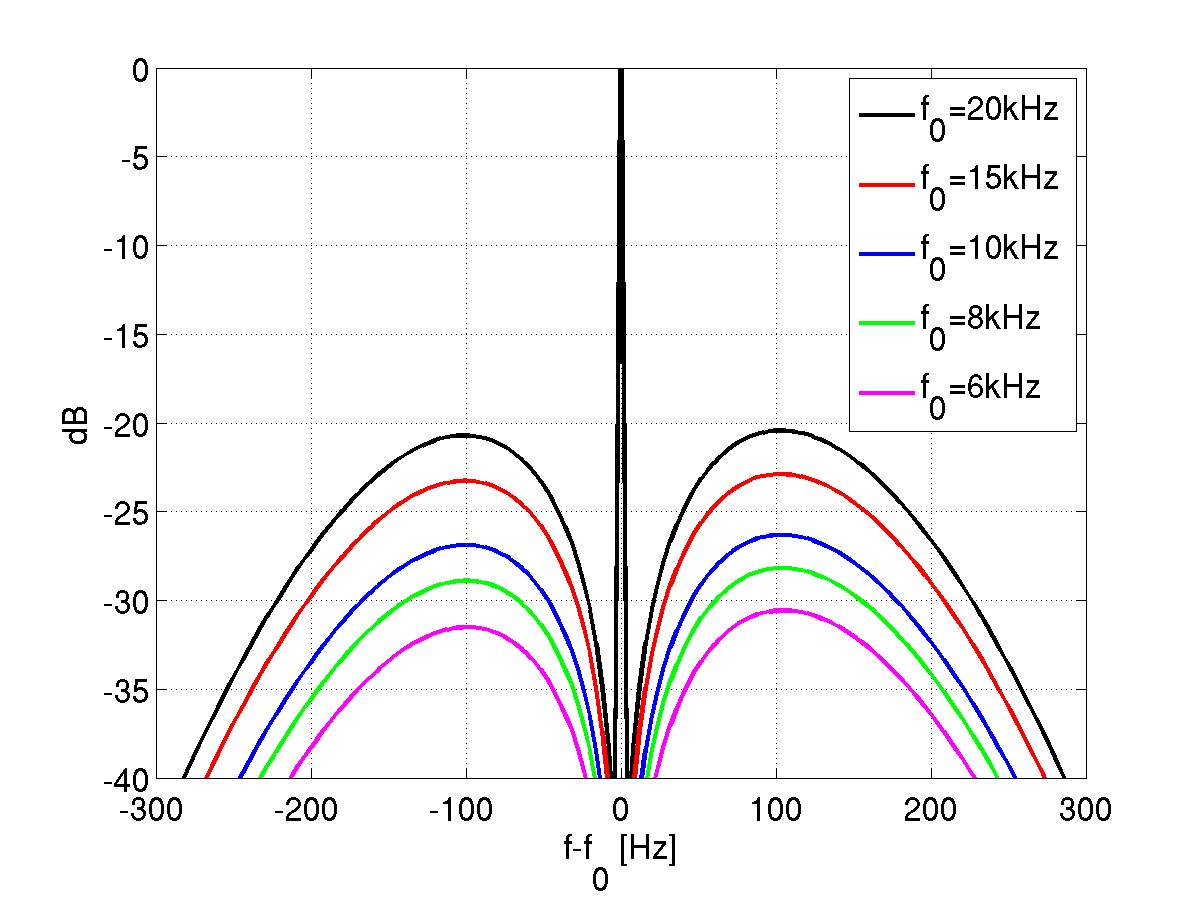}}\\
\subfigure[]{\label{3}\includegraphics[scale=0.4]{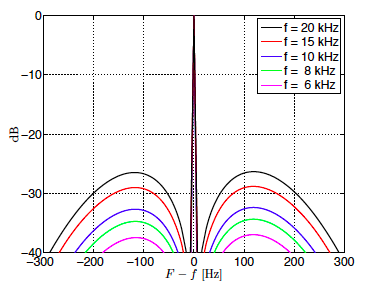}}
\subfigure[]{\label{4}\includegraphics[scale=0.4]{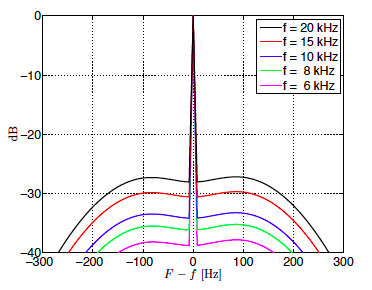}}
\caption{Normalized power spectrum of the transmitted pressure field as a function of the frequency gap $\smash{f-f_\iref}$ for various tone frequencies $\smash{f_\iref}$ in the range [6--20] kHz and jet velocity $\smash{U_J}=\norm{\vm}=60$ m/s: (a) experimental observations from \cite{CAN75}; (b) the results of \Pref{Prop:PSD:p2}; (c) the model of \cite{MCA16} for a Gaussian homogeneous axisymmetric turbulence (HAT) correlation function of the fluctuations of the ambient flow velocity; (d) the model of \cite{MCA16} for a Gaussian homogeneous isotropic turbulence (HIT) correlation function of the fluctuations of the ambient flow velocity. Note that in \cite{MCA16} the notations $\om=2\pi F$ and $\smash{\omref}=2\pi f$ are rather used.}\label{fg:DSP-vs-w0}
\end{figure}

\begin{figure}[h!]
\centering     
\subfigure[]{\label{11}\includegraphics[scale=0.48]{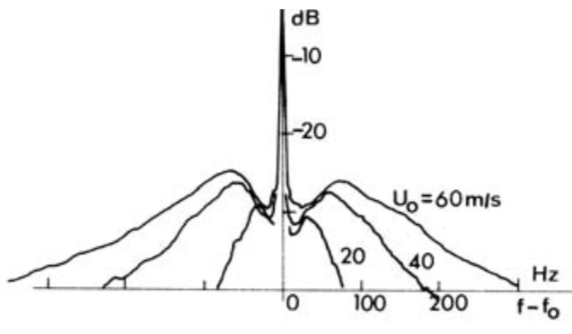}}
\subfigure[]{\label{21}\includegraphics[scale=0.27]{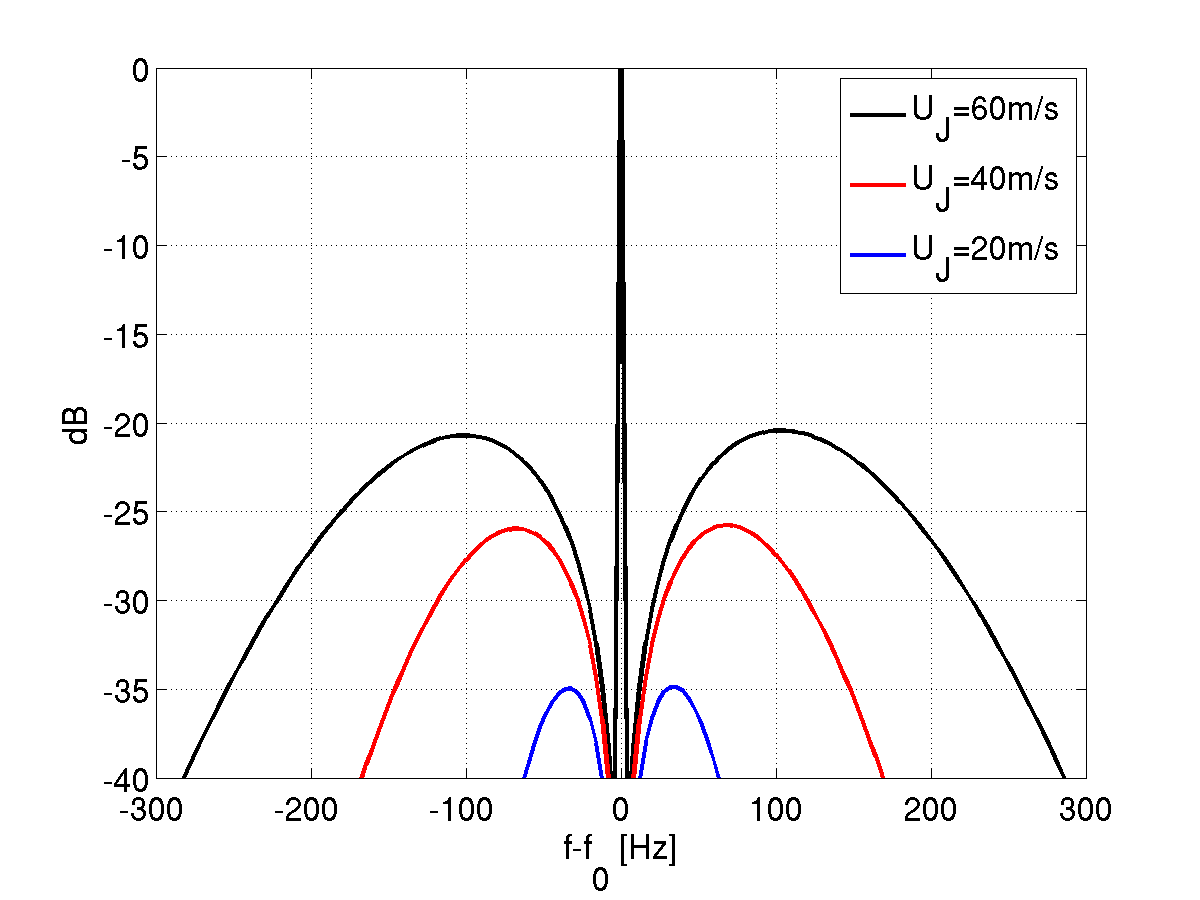}}\\
\subfigure[]{\label{31}\includegraphics[scale=0.4]{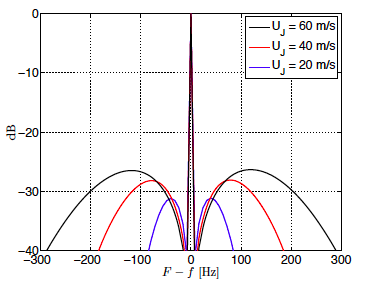}}
\subfigure[]{\label{41}\includegraphics[scale=0.4]{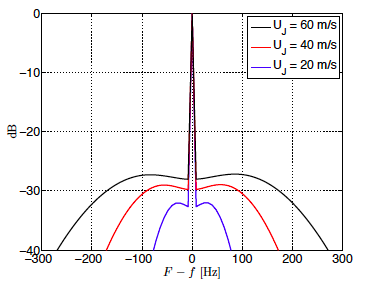}}
\caption{Normalized power spectrum of the transmitted pressure field as a function of the frequency gap $\smash{f-f_\iref}$ for various  jet velocities $\smash{U_J}=\norm{\vm}$ in the range [20-60] m/s and tone frequency $\smash{f_\iref}=20$ kHz: (a) experimental observations from \cite{CAN75}; (b) the results of \Pref{Prop:PSD:p2}; (c) the model of \cite{MCA16} for a Gaussian-HAT correlation function of the fluctuations of the ambient flow velocity; (d) the model of \cite{MCA16} for a Gaussian-HIT correlation function of the fluctuations of the ambient flow velocity. Note that in \cite{MCA16} the notations $\om=2\pi F$ and $\smash{\omref}=2\pi f$ are rather used.}\label{fg:DSP-vs-vt}
\end{figure}

\section{Summary \alert{and outlook}}\label{sec:summary}

In this paper we have developed an analytical model of the power spectral density of the acoustic waves transmitted by a plane turbulent shear layer of which ambient velocity is randomly perturbed by weak spatial and temporal fluctuations, when a time-harmonic source acts above it. The analysis starts from the linearized Euler's equations written as a Lippmann-Schwinger equation, considering that the fluctuations of the ambient flow velocity act as secondary sources for the transmitted acoustic waves.  A Born-like approximate solution of the Lippmann-Schwinger equation has been derived to work out a first-order model of the pressure field transmitted by the shear layer. In this model, the transmitted waves are constituted by their unperturbed components formed by the waves emitted by the source which have not been scattered by the ambient velocity fluctuations, and their perturbed components formed by the waves which have typically been scattered once by those fluctuations. These scattered waves are of particular interest since they have been characterized by their PSD in the experiments reported in \cite{CAN75,CAN76a,CAN76b}, which primarily motivated this work. \alert{They also motivated the linearization of Euler equations we have applied since it identifies the acoustic components of interest}. Using various assumptions for the ambient flow (thin layer) and the (high-frequency) source, and a stationary-phase argument, a model for the PSD of the scattered waves transmitted by the shear layer has been derived. It exhibits the main properties observed for the experimental PSD in \cite{CAN75,CAN76a,CAN76b,KRO13,SIJ14}, and for the numerically simulated PSD and alternative analytical models in \cite{BEN16,CLA16,EWE08,MCA09,MCA13,MCA16,POW11,SIJ14}.

The PSD of the pressure field transmitted by the shear layer shows a characteristic spectral broadening effect, whereby a reduction of the main peak at the source tone frequency in favor of more distributed spectral humps on both sides of the former, is observed. The main peak arises from the unperturbed transmitted pressure, and sidebands (lobes) arise in connection with a Doppler shift effect due to the motion of the turbulent eddies acting as secondary sources for the scattered transmitted pressure. A linear widening of these lobes with the convection velocity of the turbulent eddies has been observed, as well as the independence of the location of their maxima with respect to the tone frequency. Increasing the latter also leads to a widening of the sidebands and higher scattered levels. The proposed analysis has used a delta-correlated model of the turbulent velocity spatial fluctuations and a Gaussian model of its temporal fluctuations, though it could be improved by considering correlation models pertaining to homogeneous isotropic turbulence (HIT) or homogeneous axisymmetric turbulence (HAT) as done in \cite{MCA16}. \alert{Specifically, the general model of isotropic turbulence correlation derived in \cite{BAT53,KAR38,ROB40} could be considered in future works to possibly better match experimental observations}. However the model retained here has been able to predict the main features outlined above.

\alert{Next steps to investigate to improve our PSD model predictions could be to relax some hypotheses such as the thickness of the layer, which is assumed to be infinitely thin. Horizontally stratified flows have been considered here, but axisymmetric flow geometries have relevance to aerospace applications as well. The source-microphone axis was perpendicular to the flow in our analysis. Hence the influence of the angle of illumination could be studied as in \cite{CAM78b,CLA16,MCA13,POW11}. Besides, measurements and models of the cross power spectra between two microphones have been reported in \cite{CAN75,CAN76a,CAN76b,GUE85}, so that our model could also be compared to these results. Lastly, the velocity of the flow was assumed to be much smaller than the speed of sound. The influence of a significant Mach number could be explored further.}


\bibliographystyle{plain}

\end{document}